\pdfoutput=1
\documentclass{vldb}
\usepackage{balance}
\usepackage{graphicx}
\usepackage{subcaption}
\usepackage{booktabs} 
\usepackage{eqparbox}
\usepackage[linesnumbered,ruled]{algorithm2e}
\usepackage{color}
\usepackage{multirow}
\usepackage{euscript}
\usepackage{epsfig}
\usepackage[update,prepend]{epstopdf}
\usepackage{graphicx}
\usepackage[justification=centering]{caption}
\usepackage{mathtools}
\usepackage[labelfont=bf]{caption}
\usepackage{makecell}
\usepackage{booktabs}
\usepackage{siunitx}
\usepackage{times}
\usepackage{multicol}

\newcommand{\real}{\mathbb{R}}
\newcommand{\Emptyset}{\text{\O}}
\newcommand{\Precison}{P}
\newcommand{\Recall}{R}
\newcommand{\FSCORE}{F1}
\newcommand{\SP}{\alpha}
\newcommand{\LB}{\lambda}
\newcommand{\PP}{(1-\alpha)}
\newdef{definition}{Definition}
\usepackage{authblk}
\newtheorem{lemma}{Lemma}
\usepackage{footmisc}

\DeclarePairedDelimiterX\Basics[1](){ #1}

\vldbTitle{Efficient and Quality Seeded Graph Matching: Employing High Order Structural Information}
\vldbAuthors{{H. Zhang, Z. Huang, X. Lin, Z. Lin, W. Zhang, Y. Zhang}}

\begin{document}
\title{Efficient and High-Quality Seeded Graph Matching: Employing High Order Structural Information}

\author{
{\fontsize{12}{12}\selectfont
Haida Zhang$^{\dagger}$, Zengfeng Huang$^{\ddagger}$, Xuemin Lin$^{\dagger}$, Zhe Lin$^{\S}$, Wenjie Zhang$^{\dagger}$, Ying Zhang$^{*}$\\
}
\vspace{-3.5mm}
\fontsize{10}{10}
\selectfont
$^\dagger$University of New South Wales, Australia \ \ $^\ddagger$School of Data Science, Fudan University, China \\
\vspace{-4.3mm}
$^\S$East China Normal University, China  \ \  $^{*}$CAI, University of Technology Sydney\\
\vspace{-1mm}
\fontsize{9}{9} \selectfont\ttfamily\upshape
$^\dagger$haida.zhang@unsw.edu.au; $^\ddagger$huangzf@fudan.edu.cn; $^\dagger$\{lxue, zhangw\}@cse.unsw.edu.au; \\
\vspace{-1.3mm}
\fontsize{9}{9} \selectfont\ttfamily\upshape
$^\S$linzhe.ecnu@gmail.com; $^{*}$ying.zhang@uts.edu.au;\\
}

%
%
%
%
%


%
%


\maketitle

\begin{abstract}

Driven by many real applications, we study the problem of seeded graph matching. Given two graphs $G_1 = (V_1, E_1)$ and $G_2 = (V_2, E_2)$, and a small set $S$ of pre-matched node pairs $[u, v]$ where $u \in V_1$ and $v \in V_2$, the problem is to identify a matching between $V_1$ and $V_2$ growing from $S$, such that each pair in the matching corresponds to the same underlying entity.
Recent studies on efficient and effective seeded graph matching have drawn a great deal of attention and many popular methods are largely based on exploring the similarity between local structures to identify matching pairs. While these recent techniques work well on random graphs, their accuracy is low over many real networks. Motivated by this, we propose to utilize high order neighboring information to improve the matching accuracy. As a result, a new framework of seeded graph matching is proposed, which employs Personalized PageRank (PPR) to quantify the matching score of each node pair. To further boost the matching accuracy, we propose a novel postponing strategy, which postpones the selection of pairs that have competitors with similar matching scores. We theoretically prove that the postpone strategy indeed significantly improves the matching accuracy. To improve the scalability of matching large graphs, we also propose efficient approximation techniques based on algorithms for computing PPR heavy hitters. Our comprehensive experimental studies on large-scale real datasets demonstrate that, compared with state of the art approaches, our framework not only increases the precision and recall both by a significant margin but also achieves speed-up up to more than one order of magnitude.


\end{abstract}

\section{Introduction}\label{sec1}

In many applications such as social network de-anonymization, protein-network alignment,
pattern recognition, etc., a key task is to match two graphs $G_1=(V_1, E_1)$ and $G_2=(V_2, E_2)$ from different domains by building a mapping from $V_1$ to $V_2$. In general, $V_1$ and $V_2$ may only have a partial overlapping, thus the mapping is an injective mapping from a subset of $V_1$ to $V_2$. For example, $G_1$ and $G_2$ could be the user networks of Facebook and Twitter, where a part of the users in the two networks are the same.
Formally, given two graphs $G_1$ and $G_2$, the graph matching problem aims to identify all pairs of vertices $[u, v]\in V_1 \times V_2$ such that $u$ and $v$ correspond to the same entity (e.g., a person).

Graph match is often conducted by exploiting both structural information and semantic features \cite{malhotra2012studying, nunes2012resolving, zhang2015cosnet}. In this paper, we focus on structure-based graph matching for the following reasons:
1) the semantic features are often unavailable in the matching of networks \cite{singh2008global, mohammadi2017triangular}; 2) it has been shown structural information is the most important to identify a node in a network \cite{henderson2011s};
and 3) the techniques based on semantic features may be fragile against malicious users with fake profiles. Consequently, effective algorithms need to be developed to construct the matching relying solely on the graph structural information.
Since two graphs to be matched are usually not identical,  the problem of graph matching is much more challenging than the classic graph isomorphism problem which is widely believed to be intractable. 

In this paper, we study the problem of {\em seeded graph matching}; that is, we are additionally provided with a small set $S$ of pre-matched pairs of vertices $[u, v]$ ($u \in V_1$ and $v \in V_2$), and aim to identify a matching between $V_1$ and $V_2$ growing from $S$, such that each pair in the matching corresponds to the same  entity.



\vspace{0.1cm}
\noindent
\textbf{Applications.}
The seeded graph matching has many real applications. For example, in social networks such as Instagram and Facebook, each vertex represents a user, and an edge $(u_i, u_j)$ exists if the user $u_i$ is followed by another user $u_j$. Given that many users on Instagram may connect to their Facebook accounts, one can use such linking information as the seed pairs and identify other pairs of accounts that belong to the same individual. Consequently, we can recommend friends, social communities and products on one social network by utilizing information from other social networks. 

Another application is that seeded graph matching is a key step in a general graph matching (without seeds).  Many existing techniques \cite{zhu2013high, wang2018deepmatching, zhou2018structure} to build a general graph matching often involve two steps: 1) Seed Detection: detect two small sets $S_1$ and $S_2$ of vertices from $G_1$ and $G_2$, together with a one-to-one mapping from $S_1$ to $S_2$, and 2) Seed Propagation: expand the mapping between $S_1$ and $S_2$ to generate a graph matching between $G_1$ and $G_2$; that is, seeded graph matching.

\vspace{0.1cm}
\noindent
\textbf{Percolation Graph Matching and its Limit.}
Recently, the problem of seeded graph matching has drawn a reasonable deal of attention
\cite{yartseva2013performance, korula2014efficient, kazemi2015growing, zhou2016cross, shirani2017seeded,yacsar2018iterative}. Existing techniques is largely based on effectively using local information to identify matching pairs.  For example, if a vertex $u \in V_1$ has a neighbor $u'$ which is already matched to a neighbor $v'$ of $v \in V_2$, then $[u,v]$ gets a vote from the pair $(u',v')$ to be a valid matching pair. This is the central idea of the state of the art paradigm, {\em percolation graph matching} (PGM).
%
PGM was first proposed in \cite{narayanan2009anonymizing}; it iteratively maintains a set of matched pairs starting from the set $S$ of seed pairs: (1) in each round, each \emph{unused} (initially all seed pairs are unused) matched pair (e.g., the above $[u', v']$) percolates to its {\em neighboring} pairs (e.g., the above $[u, v]$) by adding one mark to each of them;
and (2) the pairs that have received marks, called \emph{candidate pairs}, are then examined by certain criteria and matched if they are qualified. For example, one may check whether the number of marks $r$ (known as the matching score of the pair) received is above a threshold $T$ \cite{yartseva2013performance}. Newly matched pairs are marked as unused matched pairs and will percolate in the next iteration. The algorithm terminates when all matched pairs are used (i.e., have percolated) and no new candidate pair qualifies.

There are several variants towards PGM based algorithms to specify their matching scores and matching criteria.  The authors in \cite{korula2014efficient} (also used in \cite{wang2018deepmatching}) propose to firstly consider pairs with high vertex degrees and then among such pairs, match the candidate pair with the highest score $r$ (also needs to be above a threshold T), while the authors in~\cite{kazemi2015growing} put the first preference on selecting
pairs with the highest matching score $r$ (also needs to be above a threshold T) and the second preference on minimum degree difference between the pair of vertices.  The authors in ~\cite{zhou2016cross} propose to simply choose the pairs with the maximal $r$ without the threshold constraint.
%
%

\begin{figure}[h]
\begin{centering}
\includegraphics [width=3.0in]{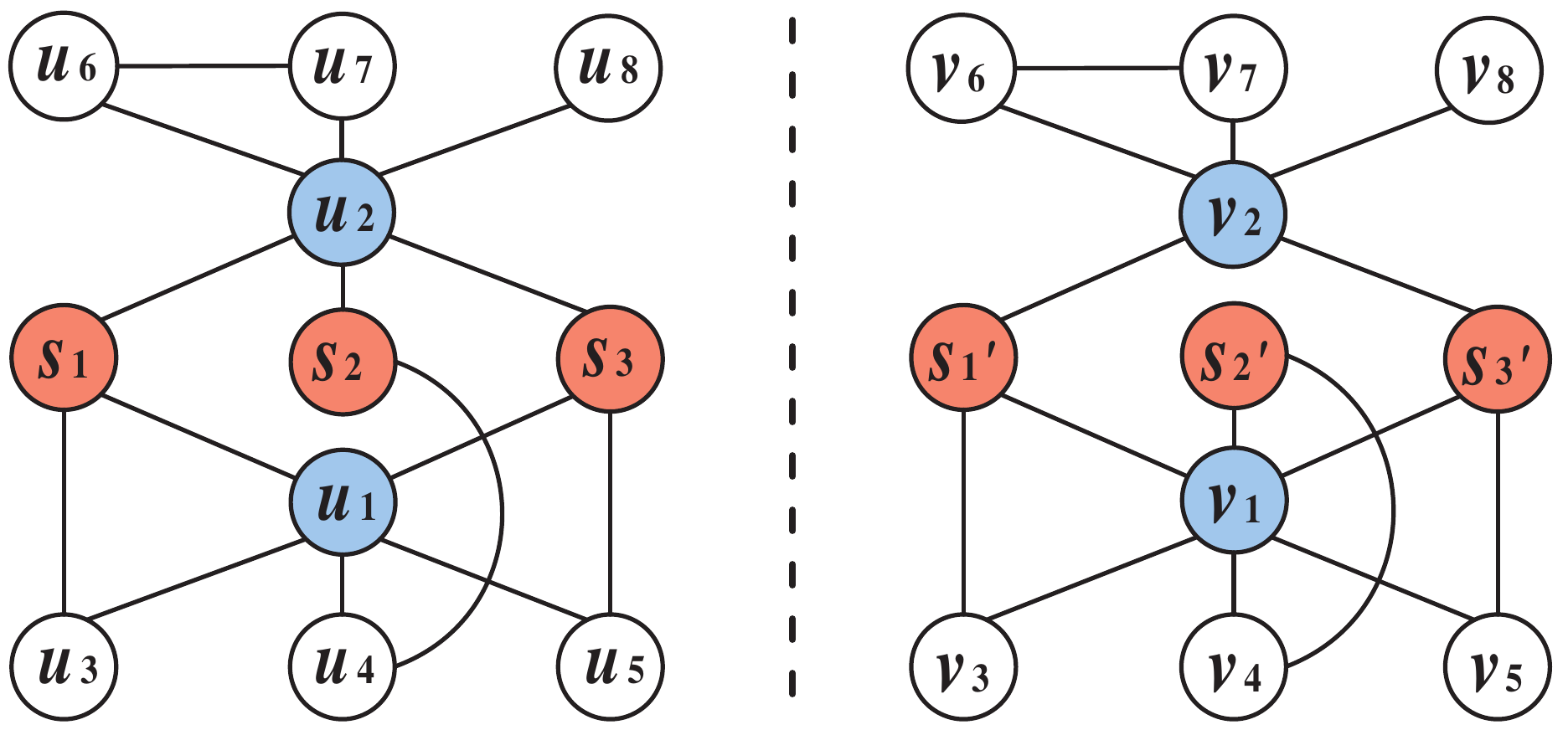}
\vspace*{-1mm}
\caption{Left Graph $G_L$ and Right Graph $G_R$. }
\label{MatchingExample}
\vspace*{-2mm}
\end{centering}
\end{figure}

A major limit of existing PGM algorithms is that relying solely on local information can easily lead to wrong matching.
Consider the example in  Figure \ref{MatchingExample}. Suppose that the seed pairs are $[s_i,s_i']$ for $1 \leq i \leq 3$. Intuitively, in this case $u_i$ should be matched with $v_i$ for $1 \leq i \leq 8$ though $G_L$ and $G_R$ are not identical.
However, the state of the art PGM algorithms \cite{korula2014efficient, zhou2016cross, kazemi2015growing} will only focus on the local information as shown in Figure \ref{local}. In particular, the algorithms \cite{korula2014efficient, kazemi2015growing} conclude that $u_1$ matches $v_2$ and $u_2$ matches $v_1$ if the threshold $T = 2$,\footnote{2 is the threshold they used.} while the algorithms in \cite{zhou2016cross} also leads to the same result by the greedy heuristic.

\begin{figure}[h]
\begin{centering}
\includegraphics [width=3.2in]{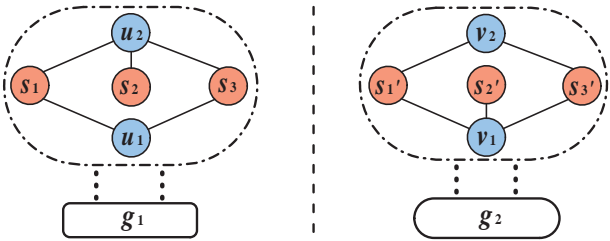}
\vspace*{-1mm}
\caption{Misled by Local Information}
\label{local}
\vspace*{-2mm}
\end{centering}
\end{figure}

\noindent
\textbf{Our approach.} To resolve this issue, we propose to evaluate the matching scores by combining (first order) neighboring information with higher order structural information. In particular, for each seed $s_i$, instead of just infecting its direct neighbors as in PGM, we propose to use a random walk-based model to quantify the relation between $s_i$ and other vertices, which is described as follows.

In the standard random walk, at each step the walker moves from the current vertex $v$ to one of its neighbors selected randomly, and each neighbor is selected with the same probability $\frac{1}{d_u}$ ($d_u$ is the degree of $u$). We will consider the decaying random walk. More concretely, the decaying random walk is the same as the standard one, except that at each step it has a probability $1-\LB$ to terminate the walk. Consider the example in  Figure \ref{MatchingExample},
starting at $s_i$, the probability that the decaying random walk reaches $u_j$ after {$1$} step is denoted as $p^1 (s_i, u_j)$. E.g.,
%
$p^1 (s_1, u_1) = \frac{1}{3} \LB$, since the walk may stop at the source $s_1$ with probability $1-\lambda$. Similarly, $p^1 (s_2, u_1) = 0$, $p^1 (s_3, u_1) = \frac{1}{3} \LB$, $p^1 (s_1, u_2) = \frac{1}{3} \LB$, $p^1 (s_2, u_2) = \frac{1}{2} \LB$, $p^1 (s_3, u_2) = \frac{1}{3} \LB$, and
$p^1 (s'_1, v_1) = \frac{1}{3} \LB$, $p^1 (s'_2, v_1) = \frac{1}{2} \LB$, $p^1 (s'_3, v_1) = \frac{1}{3} \LB$, $p^1 (s'_1, v_2) = \frac{1}{3} \LB$, $p^1 (s'_2, v_2) = 0$, $p^1 (s'_3, v_2) = \frac{1}{3} \LB$.
%
%
In general, we use $p^t (s, u)$ to denote the probability that a decaying random walk from $s$ reaches $u$ after $t$ steps. To measure the structural similarity of a vertex $u$ in $G_1$ and a vertex $v$ in $G_2$, we start a $t$-step random walk in $G_1$ ($G_2$) from a seed vertex $s_i$ ($s_i'$). We then compare the accumulated reaching probabilities from $s_i$ to $u$ and $s'_i$ to $v$ using min/max function; that is, the matching score of $(u,v)$ w.r.t. the seed pair $(s_i,s_i')$ is

\begin{equation}
\label{similarity}
\mathtt{Score}^t_i ( u, v) = \frac{\min( \sum_{k=1}^{t} p^k(s_i, u), \sum_{k=1}^{t} p^k(s'_i, v) )}{\max ( \sum_{k=1}^{t} p^k(s_i, u), \sum_{k=1}^{t} p^k(s'_i, v) )}
\end{equation}

Here, we define $\mathtt{Score}_i^t (u, v) = 0$ if both $\sum_{k=1}^{t} p^k(s_i, u)$ and $\sum_{k=1}^{t} p^k(s'_i, v)$ are $0$, so $\mathtt{Score}_i^t (u, v)\in[0,1]$.  The overall {\em matching score} of $[u,v]$ is defined as the sum of their matching scores w.r.t. all seed pairs, i.e., $\mathtt{Score}^t_{S}(u,v) = \sum_{i=1}^{|S|} \mathtt{Score}_i^t (u, v)$. Then, our algorithm will iteratively select the pair with the highest score to match.
Regarding the example in  Figure \ref{MatchingExample}, where
$S = \{ (s_1,s_1'), (s_2,s_2'), (s_3,s_3') \}$,  we have $\mathtt{Score}^1_S (u_2, v_1) = 3 $, $\mathtt{Score}^1_S (u_2, v_2) = 2 $, $\mathtt{Score}^1_S (u_1, v_1) = 2 $, $\mathtt{Score}^1_S(u_1, v_2) = 2 $, and the matching scores of all other pairs are $0$. Thus, if we stop at $t = 1$, we will choose the intuitively wrong pairs to match:
$u_2$ matches $v_1$ and then $u_1$ matches $v_2$.

However, if we continue the random walk to $t=2$, the order of scores may be reversed. One can verify that $\mathtt{Score}_S^2 (u_1, v_1) = 3 - \frac{2}{2 + \LB}$, $\mathtt{Score}^2_S (u_2, v_1) = 3 - \frac{3 \LB}{2 + \LB}$,
$\mathtt{Score}^2_S (u_1, v_2) = 2 - \frac{2}{2 + \LB}$, and $\mathtt{Score}^2_S (u_2, v_2) = 2$. Clearly, the largest matching score is either
$\mathtt{Score}^2_S (u_1, v_1)$ or $\mathtt{Score}^2_S (u_2, v_1)$. If we set $\LB$ to be larger than $\frac{2}{3}$, then $\mathtt{Score}^2_S (u_1, v_1) >  \mathtt{Score}^2_S (u_2, v_1)$; that is, $u_1$ matches $v_1$ and $u_2$ matches $v_2$. This meets the intuition.
In fact, regarding the above example, our numerical calculation also demonstrates that $\mathtt{Score}_S^\infty (u_1, v_1) > \mathtt{Score}^{\infty}_S (u_2, v_1)$ (i.e., the convergence values).

\noindent
\textbf{Contributions.}
The above is the basic idea of our algorithm and the principal conceptual contributions of our paper. We employ the \emph{Personalized PageRank} (PPR) \cite{page1999pagerank} to formalize this idea to provide a theoretic foundation towards convergence, fast algorithms, approximations, etc. The main contributions of the paper is summarized as follows.
\vspace*{-2mm}
\begin{itemize}
\item For seeded graph matching, we propose a new PPR-based score function to generate matching scores for each vertex pair.
The new paradigm directly employs high order structural information, which is general, flexible and easy to use. We also provide an analysis for the theoretical discrimination power of this match score function. Our extensive experiment results
demonstrate the new paradigm significantly improves the accuracy over state-of-the art techniques
\vspace*{-2mm}
\item We propose an optimization scheme for selecting a pair to match among all pairs sharing a common vertex. Our policy is to only match a pair which is very promising; that is, we match a pair only if its matching score is greater than $(1+\beta)$ times the matching scores of all its competitors (which share a common vertex with the pair) for some $\beta>0$. While our experiment shows it can bring significant boost in accuracy, we also theoretically prove that this scheme can improves the accuracy for a random graph model.
\vspace*{-2mm}
\item To improve the scalability of our framework, we develop efficient and effective approximation techniques based on approximation algorithms for PPR heavy hitters~\cite{wang2018efficient}, i.e., given a source node, the algorithm only reports the PPRs that are relatively large. Such techniques enable us to speed-up the matching computation by more than one order of magnitude while still retain the significant accuracy boost compared with state-of-the-art algorithms.
\vspace*{-2mm}
\item We conduct extensive experimental studies on various large-scale real-world graphs, which demonstrate that our new graph matching framework outperforms start-of-the-art algorithms regarding accuracy, robustness and efficiency.
\end{itemize}
\vspace*{-1mm}

\noindent
\textbf{Outline.} The related work immediately follows. Section 2 introduces notations and PPR.
Section 3 proposes a new matching score function and a basic graph matching algorithm.
Section 4 presents our PPRGM framework with unmatched advantages over the start-of-the-art methods, as well as effective optimization techniques to speed-up the computation.
It also presents the analysis of time complexity and effectiveness of our algorithm.
Section 5 evaluates all introduced algorithms using extensive experiments.
Finally we conclude the paper in Section 6.



\noindent
\textbf{Related Work.}
%
The graph matching problem is a generalization of the classic isomorphic mapping problem that is intractable in general \cite{lewis1983computers}. It has been widely used as a building block in various applications.
These include the alignment of protein-protein interaction networks in systems biology \cite{singh2008global, klau2009new, mohammadi2017triangular}, an identification of users (e.g., username, description, location and profile image) across different communities \cite{zafarani2009connecting, yartseva2013performance, korula2014efficient} in social networks, and de-anonymization which breaks the privacy of social networks \cite{backstrom2007wherefore, narayanan2009anonymizing, mislove2010you, wondracek2010practical, pedarsani2011privacy, ji2016seed}.

With the assumption of non-priori knowledge of the alignment, conventional graph matching algorithms firstly compute the similarity of each pair (all similarities are represented by a $|V_1|\times |V_2|$ matrix) and then identify the alignments with high scores.
\cite{singh2008global} builds the similarity matrix by iterative weighted propagation.
\cite{zhu2013high} consider both global vertex similarity $S_g$, defined by the similarity of corresponding eigenvectors from the spectral clustering, and local vertex similarity $S_{\ell}$, defined by the similarity of degree sequences of 2-neighborhood subgraphs.
Recent works \cite{wang2018deepmatching, zhou2018structure} propose to use deep-walk based graph embedding to compute the low-rank representation of each vertex. Then \cite{zhou2018structure} builds the node-to-node similarity matrix by defining a normalized Euclidean distance, while \cite{wang2018deepmatching} adopts Coherent Point Drift (CPD) method \cite{myronenko2007non} to identify the seeds and use \cite{korula2014efficient} to conduct seeded graph matching.

The seeded graph matching problem studied in this paper, however, assumes a small set of ``pre-matched'' pairs exists at the beginning. As discussed earlier,
the seed matches can be obtained by public link information \cite{zhang2015cosnet},
or by existing seed identification methods \cite{singh2008global, zhu2013high, wang2018deepmatching}.
A popular class of works are Percolation Graph Matching (PGM) \cite{korula2014efficient,yartseva2013performance,zhou2016cross}. 
While PGM is scalable and has a reasonable precision, traditional PGM methods may suffer from early-stop (low recall), especially when only a very small number of seeds are provided.
To prevent early termination, recently, ~\cite{kazemi2015growing} proposes a new strategy, namely ExpandWhenStuck (EWS).
Deviating from PGM, more recently, \cite{yacsar2018iterative} proposes to use seeds as anchor points to locate other vertices in a $2$-D space and then conduct mapping for the vertices in a same region. As we showed in
the introduction, the recent techniques in \cite{korula2014efficient,kazemi2015growing,yacsar2018iterative,zhou2016cross}
are all limited by solely looking at local information.

In approximate isomorphic mapping, it has been focused on maximizing the number of matched edges
\cite{singh2008global, klau2009new, zhu2013high}; and  \cite{mohammadi2017triangular} considers to maximize the number of matched triangles. A more related work \cite{zhu2013high} to ours is to first generate anchors and then extend the anchors to identify an approximate mapping.
It is quite slow as is demonstrated by our experiments. 

In our experiments, we compare our techniques with the state-of-the-art techniques \cite{kazemi2015growing,yacsar2018iterative,zhou2016cross,zhu2013high}.
The extensive experiments demonstrate that our algorithms significantly outperform these techniques. We did not compare \cite{korula2014efficient} since the authors in
 \cite{kazemi2015growing} already showed the advantage of EWS against \cite{korula2014efficient}.

%

\section{Preliminaries} \label{sec2}

\begin{table}[!t]
\renewcommand{\arraystretch}{1.1}
\centering
\caption{\label{table:notation}\bf{Notations and Descriptions}}
\vspace*{-1mm}
\resizebox{\columnwidth+0.1in}{!}{%
\begin{tabular}{l|l}
\noalign{\hrule height 1pt}
\textbf{Notation} & \textbf{Descriptions}\\
\hline
\hline
$G_1(V_1,E_1)$,  & The graphs $G_1$ and $G_2$ to be matched, with vertex  \\
    $G_2(V_2,E_2)$   &   sets $V_1$ and $V_2$, and edge sets $E_1$ and $E_2$. \\

\hline
$(u,u')$     & The edge between vertex $u$ and vertex $u'$.\\
\hline
$[u,v]$     & The vertex pair where $u\in V_1$ and $v\in V_2$.\\
\hline
$N(u)$      &  The neighbor set of vertex $u$.\\
\hline
$N(u,v)$    &  The neighboring pair set of $[u,v]$.\\
\hline
$\SP$    &  The probability that a random walk terminates \\
            & at the current vertex.\\
\hline
$\pi(s,u)$  &  The exact PPR value of $u$ with respect to $s$.\\
\hline
$R(u)$      &  The vector of PPR values of $u$ with respect to  \\
& all seed vertices.\\
\hline
$\mathtt{Score}_k(u,v)$ &  The matching score of a pair $[u,v]$ with respect  \\
            & to the $k$-th seed pair $[s_k,s'_k]$.\\
\hline
$\mathtt{Score}_S(u,v)$ &  The matching score of a pair $[u,v]$ with respect   \\
            & to the set $S$ of all seed pairs.\\
\hline
$r_{max}$   & The residue threshold for Forward-Push.\\
\hline
$r(s,u)$    &  The residue of $u$ during a Forward-Push from $s$.\\
\hline
$\pi^o(s,u)$&  The reserve of $u$ during a Forward-Push from $s$.\\
\hline
$H(s)$     &  The set of PPR heavy hitters obtained from a\\
            &  Forward-Push from $s$.\\

\noalign{\hrule height 1pt}
\end{tabular}
}
\vspace*{-3mm}
\end{table}

\subsection{Graph Notation} \label{21notation}
Notations used in this paper is summarized in Table~\ref{table:notation}. The explanation of some notations are provided here.
Consider two undirected graphs $G_1(V_1, E_1)$ and $G_2(V_2, E_2)$ to be matched.
We use $(u,u')\in E_{1} (E_2)$ to denote an edge between two $u,u'$, and $[u,v]$ to represent a pair of vertices where $u \in V_1$ and $v\in V_2$. $N(u)$ denotes the neighbors of vertex $u$. The neighboring pairs of a pair $[u,v]$ is the set of all pairs $[u', v']\in V_1 \times V_2$ such that $(u, u')\in E_1$ and $(v, v') \in E_2$, which is exactly $N(u)\times N(v)$. We use $N(u,v)$ to denote the set of neighboring pairs of $[u,v]$.  The initial seed set is $S=\{[s_1, s'_1], [s_2, s'_2], ..., [s_{|S|}, s'_{|S|}]\}$, where $s_k \in V_1$ and $s'_k \in V_2$ for $1\le k\le |S|$.

An undirected graph $G$ can be represented by an adjacency matrix $A$, where $A_{ij} =1$ if $(i,j)\in E$ and $A_{ij}=0$ otherwise. Note that $A$ is symmetric since $G$ is undirected.
The probability transition matrix $P$ defined on $G$ is given by
\[
    P_{ji}=
\begin{cases}
    \frac{1}{|N(i)|} & \text{if } (i,j)\in E,\\
    0              & \text{otherwise.}
\end{cases}
\]
We use $[\boldsymbol{x}]_i$ to denote the $i^{th}$ value of the vector $\boldsymbol{x}$.

\subsection{Personalized PageRank} \label{22PPR}
As motivated in the introduction, we will employ Personalized PageRank (PPR) to quantify the connections between the matched vertices and unmatched vertices.
PPR has been widely adopted in graph structure analysis, such as community detection~\cite{avron2015community}, graph ordering~\cite{benson2016higher, yin2017local}, and other applications~\cite{gleich2015pagerank}.

Given a source vertex $s\in V$ and a stopping probability $\SP$, 
a \emph{decaying random walk} is a traversal of $G$ that starts from $s$ and, at each step:
(1) with probability $\PP$, proceeds to a randomly selected neighbor of the current vertex, or
(2) with probability $\SP$, terminates at the current vertex.
For any vertex $u\in V$, its PPR value $\pi(s,u)$ w.r.t. source $s$ is the probability that a decaying random walk from $s$ terminates at $u$.

Starting from vertex $s$, let $q_{su}^{(t)}$ be the probability that the (non-decaying) random walk reaches vertex $u$ after $t$ steps. Since the probability that a decaying random walk doesn't terminate before step $t$ is $\PP^t$, the probability that the decaying random walk reaches vertex $u$ in the $t^{th}$ step is
$p_{su}^{(t)} = \PP^t \cdot q_{su}^{(t)}.$ Therefore,
$$\pi(s,u) = \alpha \cdot \sum_{t=0}^\infty \PP^t \cdot q_{su}^{(t)} $$

Such probabilities can be represented in matrix forms. Let $e_{s}\in \real^{|V|\times1}$ be the  $s^{th}$ standard basis vector, i.e. with $1$ at the $s^{th}$ position and $0$'s everywhere else. 
Let $P$ be the transition matrix defined in Section~\ref{21notation}, then
\[
q_{su}^{(t)} = [P^{t}\cdot e_{s}]_u, \ \  \textrm{and} \ \  p_{su}^{(t)} = [\PP^{t}\cdot P^{t}\cdot e_{s}]_u.
\]
Hence, the PPR value of vertex $u$ w.r.t. source $s$ is
\begin{equation}\label{ep_PPR}
\begin{split}
\pi(s, u) & = \SP\cdot\sum_{t=0}^{\infty}p_{su}^{(t)}= \SP \cdot \sum_{t=0}^{\infty}[{\PP}^{t}\cdot P^{t}\cdot e_{s}]_u\\ 
                                                      &=[(I-{\PP}P)^{-1}\cdot (\SP e_{s})]_u,
\end{split}
\end{equation}
where $I\in \real^{|V|\times |V|}$ is the identity matrix.
Computing exact PPR is time consuming, so approximation algorithms are widely studied and used in practice.

\section{Basic Algorithm} \label{Section3}

In this section, we present the basic idea of our graph matching algorithms, which will be refined in the following sections. 

\subsection{Matching Score Function} \label{31score}
The main ingredient is a new matching score function, which quantify the credibility to match any vertex pair.
Given a set $S$ of seed pairs,
we define the \emph{signature vector} of vertex $u\in V_1$ as $R(u)= \{\pi(s_1, u), \pi(s_2, u), ..., \pi(s_{|S|}, u) \}$,
and  
that of vertex $v\in V_2$ as $R(v)= \{\pi(s'_1, v), \pi(s'_2, v), ..., \pi(s'_{|S|}, v) \}$, where $[s_k,s'_k]\in S$ for $1\le k \le |S|$.

Intuitively, the closer the PPR values of $u$ and $v$ are (with respect to a seed pair $[s,s']$), the more likely $[u,v]$ is to be a valid match.
Obviously, on two isomorphic graphs $G_1$ and $G_2$, $\pi(s,u)=\pi(s',v)$ always holds provided that $[s,s']$ and $[u,v]$ are two valid matches.
%

Thus we propose to define the matching score of a pair $[u,v]$ with respect to the $k$-th seed pair $[s_k,s'_k]$ by
\begin{equation}\label{ep_score1}
\mathtt{Score_k}(u,v) = \frac{\min{(\pi(s_k, u),\pi(s'_k, v))}}{\max{(\pi(s_k, u),\pi(s'_k, v))}}.
\end{equation}
We let $\mathtt{Score_k}(u,v)=0$ if both $\pi(s_k, u)$ and $\pi(s'_k, v)$ are $0$.


The matching score of $[u,v]$ with respect to the seed set $S$ is simply the summation of scores over all seed pairs:
\begin{equation}\label{ep_score2}
\mathtt{Score}_S(u,v) = \sum_{k=1}^{|S|} \mathtt{Score_k}(u,v).
\end{equation}
This utilizes all the seed information and makes the score function more robust against possible wrong matches in $S$.

\subsection{Basic Graph Matching Algorithm}\label{32basic}

%
%
%
%
%

Based on the proposed matching score function, we give a basic graph matching algorithm which greedily matches pairs with highest scores. 
For each seed pair $[s, s']$, the algorithm marks $s$ and $s'$ as matched;
two PPR computations are conducted on $s$ and $s'$ respectively with the same stopping probability $\SP$, which
return $\pi(s,u)$ for all $u\in V_1$ and $\pi(s',v)$ for all $v\in V_2$. 
Next, for each pair $[u,v]\in V_1\times V_2$,  the algorithm computes $\mathtt{Score}_S(u,v)$ based on~\eqref{ep_score1} and~\eqref{ep_score2}.
Initially all possible pairs, except for those containing vertices that have been marked as matched, are inserted into the candidate set $C$.
Then, the algorithm matches pairs iteratively: at each step, it greedily picks a pair $[u,v]\in C$ with the largest $\mathtt{Score}_S(u,v)$,
and then removes all pairs in $\{u\}\times V_2 \cup V_1\times\{v\}$ from $C$; $u$ and $v$ are marked as matched.
The algorithm terminates when $C$ becomes empty.

\vspace{0.1cm}
\noindent
\textbf{Limits of basic graph matching.}
The above naive algorithm has two critical issues.
Firstly, it has high computational costs and thus not scalable: computing all PPR values from all seed vertices is time-consuming;
even if the PPR values are given for free, computing matching scores for all $[u,v]\in V_1\times V_2$ takes $O(|S|\cdot |V_1|\cdot|V_2|)$ time.
Secondly, the score function of a pair $[u,v]$ only depends on the seed pairs, which is potentially insufficient for computing high-quality matching when the number of initial seeds is very limited.

\section{PPRGM Algorithm}\label{4approx}



To improve the time efficiency, we do not compute all PPR values and the matching scores in the beginning. Instead, PPRGM adopts a greedy expansion mechanism similar to the percolation process in PGM to generate candidates iteratively; only the matching scores of these candidate pairs will be computed. During this process, candidate pairs that satisfy certain criteria are marked as \emph{matched}. To resolve the second issue mentioned above, in addition to the initial seeds, the matching score of a candidate pair will be augmented by the information of early matches. More precisely, the matching score of a candidate pair is now the summation of the score w.r.t. initial seeds and the score w.r.t. earlier matches that are added by the algorithm.


%
%
%

\subsection{Main Framework}\label{32PPRGM}
To generate candidates, we use a similar idea as in PGM. Given any matched pair $[s,s']$ (either a seed pair or an early match), the hypothesis adopted by PGM is that the neighboring pairs of $[s,s']$ are more likely to be valid matches. So, in each iteration, PGM simply adds the neighboring pairs of any newly matched pair as candidates.  
In our framework, we propose more general and effective strategies for generating candidate pairs, which is also based on the idea of PPR.
The matching score of a candidate pair $[u,v]$ w.r.t. each matched pairs $[s,s']$ will be computed according to~\eqref{ep_score1}. Let $M$ be the set of all the initial seed pairs and the pairs matched by the algorithm, then the matching score of $[u,v]$ is computed by $\mathtt{Score}_M(u,v) = \sum_{k=1}^{|M|} \mathtt{Score}_k(u,v)$.
\begin{algorithm}[!tb]
\caption{Personalized PageRank based Graph Matching (PPRGM)}\label{alg:PPRGM}
    \SetKwInOut{Input}{Input}
    \SetKwInOut{Output}{Output}
    \Input {$G_1(V_1, E_1), G_2(V_2, E_2)$, set $S$ of seed matches, stopping probability $\SP$.}
    \Output {The set of matched pairs $M$.}

    Let $C\leftarrow\Emptyset$ store the candidate pairs; \\ 

    Let $M'\leftarrow\Emptyset$ store the pairs that are matched but have not been added to $M$;

    \ForEach {seed pair $[s_k, s'_k]\in S$}{
       Insert $[s_k, s'_k]$ into $M'$; mark $s_k, s'_k$ as matched;\\
    }

    \While{$True$}{
        \ForEach {pair $[u,v]\in M'$}{
            Insert $[u,v]$ into $M$;\\
            \textbf{Candidate-Set-Expansion($u, v, C$)};\\
    	}

        $M'\leftarrow\Emptyset$. \\
        \ForEach {$[u,v]\in C$ satisfying matching criteria}{
                Insert $[u,v]$ into $M'$; mark $u, v$ as matched;\\
                Remove pairs in $\{u\}\times{V_2}$ and ${V_1}\times \{v\}$ from $C$;\\
        }
        \If {$M'$ is empty}{
            Relax the matching criteria, or break the while-loop if the matching criteria cannot be further relaxed;\\
        }
    }
    \Return {M}

\end{algorithm}

The main PPRGM framework is presented in Algorithm~\ref{alg:PPRGM}.
Given two graphs $G_1(V_1,E_1)$ and $G_2(V_2,E_2)$ and a set $S$ of seed matches.
The algorithm initializes a set $C$ to store the candidate pairs, 
and a set $M'$ to store the pairs that are matched but have not been added to $M$ (line 1,2).
Firstly, PPRGM inserts each seed pair $[s_k,s'_k]\in S$ into $M'$, and marks $s_k$ and $s'_k$ as matched (line 3-5).
Then PPRGM starts a while-loop which iteratively expands the candidate set and matches vertex pairs: 
(1) For each pair $[u,v]\in M'$, the algorithm inserts $[u,v]$ into $M$, and calls a \emph{Candidate-Set-Expansion} procedure, which expands the candidate set $C$ by adding all neighboring pairs of $[u,v]$ to $C$ (which will be further discussed in Section~\ref{43candexpansion});
$M'$ is then set to empty (line 7-11);
(2) For each pair $[u,v]\in C$ that satisfies certain matching criteria, PPRGM inserts $[u,v]$ into $M'$ and marks $u$ and $v$ as matched, then the candidate pairs that belong to $\{u\}\times{V_2}$ or ${V_1}\times \{v\}$ are removed from $C$ (line 12-15). In our framework, the matching criteria are parameterized by two parameters $\gamma, \beta$, which will be discussed shortly. 

To reduce the chance of matching wrong pairs in the early stages, we introduce the \emph{postponing strategy}. Roughly speaking, we will apply more strict matching criteria in the early stages so that only very promising pairs are selected and the decisions that whether to match some ``uncertain" candidate pairs are postponed. The matching criteria will be relaxed periodically so that more pairs can be matched (line 16-18).
\subsection{Matching Criteria and Postponing Strategy}
We define \emph{adversary pairs} of $[u,v]$ as the set of all pairs in $\{u\}\times V_2$ and $V_1\times \{v\}$ excluding $[u,v]$, which are the direct competitors of $[u,v]$. Note that a pair will not be matched if any one of its adversary pairs is matched first. A simple idea is then to match the pairs whose adversary pairs in the candidate set all have smaller matching scores.

Suppose $[u,v]$ is a correct pair. Since the candidate pairs are generated and matched iteratively, the following \emph{undesirable cases} may happen:

(1) $[u,v']\in C$ but currently $[u,v]\notin C$, although $\mathtt{Score}(u,v)>\mathtt{Score}(u,v')$; 

(2) $[u,v]\in C$ and $[u,v']\in C$, but currently $\mathtt{Score}(u,v')$ is slightly higher than $\mathtt{Score}(u,v)$.

To reduce the chance of making wrong decisions in these cases (especially in the early stages), we adopt a  \emph{postponing} strategy, which aims to postpone the matching of $uncertain$ pairs.
We call $[u,v]$ a \emph{$(\gamma,\beta)$-strong pair}\footnote{if a pair is not strong, then it is uncertain}, if:
(1) $\mathtt{Score}(u,v)> \gamma$ where $\gamma$ is a predefined score threshold;
(2) $[u,v]$ has no $\beta$-close adversary candidate pairs, where $\beta$-close adversary is defined as follows.

\begin{definition}
	For any score function $\mathtt{Score}(,)$ and some $\beta>0$, we say a candidate pair $[u',v']$ is a $\beta$-close adversary pair of another candidate pair $[u, v]$, if $\mathtt{Score}(u,v)\le (1+\beta)\cdot \mathtt{Score}(u',v')$ and $[u',v']$ is an adversary pair of $[u,v]$.
\end{definition}
Then we say a pair $[u,v]$ has no $\beta$-close adversary (candidate) pairs,
if $\mathtt{Score}(u,v)>(1+\beta)\cdot \mathtt{Score}(u',v')$ for all adversary pairs $[u',v']$ of $[u,v]$ that are also in the candidate set $C$.



If a pairs is not strong then it is uncertain.
In each iteration, \emph{the algorithm postpone the matching of any uncertain pair, even if it has the highest matching score, and only matches strong pairs in the candidate set.} One should observe that, using this strategy, there will be no tie. If $\gamma$ and $\beta$ are set to be large, we may avoid wrongly matching pair $[u,v']$ in the above two undesirable cases.
On the other hand, after PPRGM adds more pairs to $M$ and expands the candidate set, the matching priorities of $[u,v]$ and $[u,v']$ might be reversed, and the correct match would be identified.

We employ strict matching criteria by assigning $\beta$ and $\gamma$ with relatively large values in the beginning (Empirically, we set $\beta=1$ and $\gamma=|S|/2$) to only match the most ``certain" pairs. We will relax the criteria by decreasing $\beta$ and $\gamma$ ($\beta=\beta/2$, $\gamma=(\gamma+1)/2$), when there are currently no vertex pair in $C$ that satisfies the matching criteria.

The effectiveness of the postponing strategy will be analyzed at the end of this section.

\subsection{PPR Heavy Hitters and Forward-Push}\label{42expansion}
Recall that, in PPRGM, each time a pair $[u,v]$ is inserted to $M$ (i.e., matched), its neighborhood will be added to the candidate set. Computing the score function of a new candidate requires computing $|M|$ PPR values, which is too high when $M$ becomes large. Note that, typically, $|M|$ could be $O(|V|)$, so even if a single PPR query can be done in $O(1)$ time, the total running time is $O(|V|\cdot |M|) = O(|V|^2)$.

On the contrary, we precompute the PPR values. More specifically, every time a new pair $[u, v]$ is added to $M$, our algorithm conducts a single source PPR computation for $u$ and $v$ respectively, which computes the PPR values for all vertices; the PPR values are then stored in the memory. However this still takes too much time and space. As a single source PPR computation takes at least $O(|V|)$ time, the total time is still $O(|M|\cdot |V|) = O(|V|^2)$, and we also need so much space to store all the PPRs.

Our solution is to only compute and store the ``heavy hitters" of all PPR values, i.e., PPRs that are relatively large. Computing the heavy hitters can be much faster than a full single-source PPR computation~\cite{wang2018efficient}. Hence, in PPRGM, when a pair $[u, v]$ is added to $M$, we compute and store the PPR heavy hitters of $u$ and $v$, and regard all non-heavy PPRs as $0$. This approximate approach is time and space efficient, and produces high-quality matchings as shown in our experimental studies.

\begin{algorithm}[!tb]
\caption{Forward-Push method}\label{alg:FP}
    \SetKwInOut{Input}{Input}
    \SetKwInOut{Output}{Output}
    \Input {$G(V, E)$, source vertex $s$, stopping probability $\SP$, residue threshold $r_{max}$}
    \Output {$H(s)$ containing PPR heavy hitters and their approximate PPR values} 

    set $r(s, s)$ $\leftarrow$ $1$, and $r(s,u)$ $\leftarrow$ $0$ for all vertices $u\ne s$;

    set $\pi^o(s, u)$ $\leftarrow$ $0$ for all $u\in V$.

    \While{$\exists u\in V$ such that $r(s,u)/|N(u)| > r_{max}$}{
        $\pi^o(s,u)$ $\leftarrow$ $\pi^o(s,u)+\SP\cdot r(s,u)$;

        \ForEach {$v\in N(u)$}{
            $r(s, v)$ $\leftarrow$ $r(s,v) + \PP\cdot \frac{r(s,u)}{|N(u)|};$
        }

        $r(s,u)$ $\leftarrow$ $0;$
    }
    \Return {$H(s)=\{u|\pi^o(s,u)>0, u\in V\}$}
\end{algorithm}

\vspace{0.1cm}
\noindent
\textbf{Forward-Push.}
Although more efficient algorithms exist for computing PPR heavy hitters~\cite{wang2018efficient}, we use the classic Forward-Push algorithm of~\cite{andersen2007local}. This is because Forward-Push is a local algorithm (only explores a local neighborhood of the source vertex), which is simple, efficient and deterministic. Moreover, the accuracy and running time of Forward-Push are controlled by a residue threshold $r_{max}$, which is suitable for our purpose. Given a  threshold $r_{max}$, Forward-Push returns an approximate heavy hitter set together with their approximate PPR values. Next we give a brief overview of the Forward-Push algorithm.

The pseudo code is presented in Algorithm~\ref{alg:FP}.
It starts a PPR approximation on source vertex $s$.
Given a graph $G(V,E)$, the source vertex $s$, a stopping probability $\SP$, and a residue threshold $r_{max}$. The Forward-Push algorithm defines two variables on each vertex in graph all the time: (1) the reserve $\pi^o(s,u)$ represents the current estimation of the probability that a random walk terminates at vertex $u$ (i.e. the PPR value), which is always an underestimation of $\pi(s,u)$; (2) the residue $r(s,u)$ represents the probability that the random walk is currently at vertex $u$.
Note that $\sum_{u\in V}{\pi^o(s,u)}+\sum_{u\in V}{r(s,u)}=1$ holds at all times (the intuitive explanation is that the random walk either has terminated at some vertex or is currently at some vertex).
Initially $r(s,s)$ is set to $1$ while $\pi^o(s,u)=0$, $\forall u\in V$, and the algorithm iteratively increases $\sum_{u\in V}{\pi^o(s,u)}$ to have better approximations to the PPR values.
In each iteration, the algorithm finds a vertex $v$ with $r(s,u)/|N(u)|>r_{\max}$\footnote{When there are multiple vertices satisfying the condition $r(s,u)/|N(u)|>r_{\max}$, the algorithm always picks the vertex $u$ with \emph{maximum} $r(s,u)/|N(s)|$, thus the number of iterations is greatly reduced as each iteration makes the most progress.} (line 3), then distributes the value of $r(s,u)$ to $u$'s neighbors: (1) it increases the reserve $\pi^o(s,u)$ by $\SP \cdot r(s,u)$ (line 4); (2) it adds $\PP \cdot \frac{r(s,u)}{|N(u)|}$ to each residue $r(s,v)$ for each $v\in N(u)$ (line 5-7); (3) $r(s,u)$ is set to $0$.
The Forward-Push algorithm terminates when $r(s,u)/|N(u)|\le r_{\max}$ for all $u\in V$ simultaneously.
A set $H(s)$ is returned containing the PPR heavy hitters together with their approximate PPR values.

To sum up, each time a pair $[u,v]$ is added to the matching set $M$, PPRGM invokes Forward-Push to compute and store  PPR heavy hitters of $u$ and $v$, and then expands the candidate set accordingly.

\begin{procedure}[!tb]
\caption{3: High-Order-Expansion($u, v, C$)}\label{proc:HOE}
    $H(u)$  $\leftarrow$  Forward-Push$(G_1, u, \SP, r'_{max})$;\\
    $H(v)$ $\leftarrow$  Forward-Push$(G_2, v, \SP, r'_{max})$;

	\ForEach {pair $[u',v']\in H(u)\times H(v)$}{
    	\If {either $u'$ or $v'$ is matched}{
            Continue;
        }
        $\mathtt{sc}$ $\leftarrow$ $\frac{\min{(\pi^o(u, u'),\pi^o(v, v'))}}{\max{(\pi^o(u, u'),\pi^o(v, v'))}+\sigma}$;\\

        \eIf {$[u',v']\in C$}{
            $\mathtt{Score}(u',v')$ $\leftarrow$  $\mathtt{Score}(u',v')$ + $\mathtt{sc}$;
        }
        {
            $\mathtt{Score}(u',v')$ $\leftarrow$ $\sum_{k=1}^{|S|} \mathtt{Score_k}(u',v')$ + $\mathtt{sc}$;\\
            Insert $[u',v']$ into set $C$;\\
        }

    }
\end{procedure}
%
%
%
%

\subsection{Candidate Set Expansion}\label{43candexpansion}
Let $H(u)$ and $H(v)$ be the PPR heavy hitter set of $u$ and $v$ (returned by Forward-Push) respectively.  PPRGM will add all pairs in $H(u)\times H(v)$, which haven't been added before by other matched pairs, to the candidate set $C$. If a pair in $H(u)\times H(v)$ has already existed in $C$, its matching score is updated according to the PPRs w.r.t. $u$ and $v$. Otherwise, we initialize the match score of this pair.

We consider the initial seed pairs as the most valuable information, and thereby treat seed pairs differently from matched pairs added by the algorithm. In particular, the matching score of a pair $[u,v]$ is divided into two parts. The first part is the score w.r.t. seed pairs,  for which we compute the PPR values with higher accuracy (smaller residue threshold in Forward-Push); the second part is w.r.t. other matched pairs, for which we use lower-accuracy PPR values, which will speedup the computation significantly.

More specifically, at the beginning of the algorithm, we conduct a Forward-Push on each seed vertex $s$ ($s'$) with a relatively small $r_{\max}$. All the PPR values w.r.t. seed vertices are stored (using $O(|S|/r_{\max})$ space).
The basic matching score of a pair $[u,v]$ w.r.t. all initial seed pairs is $\sum_{k=1}^{|S|} \mathtt{Score_k}(u,v)$. However, to save space and time, the algorithm will not compute and store the basic matching score of $[u,v]$ until $[u,v]$ is added to $C$ as a candidate pair.
When a new pair is matched (added to $M$) during the algorithm, we conduct a Forward-Push for each of the two end vertices with a relatively large $r'_{max}$ to compute their PPR heavy hitters and then add the heavy hitters to $C$.
Since in general $H(v)$ contains a high-order neighborhood of $u$, we call this candidate set expansion mechanism as \emph{High-Order (Neighbor) Expansion}.

\vspace{0.1cm}
\noindent
\textbf{The High-Order Expansion and Neighbor Expansion.} We present the details of High-Order Expansion (HOE, see Procedure 3).
For each matched pair of $[u,v]\in M'$, the algorithm conducts Forward-Push procedures on $u$ and $v$ respectively,
using the same $\SP$ and the same residue threshold $r'_{max}$ (line 1-2).
The Forward-Push procedures return $H(u)$ and $H(v)$, which contain the PPR heavy hitters of $u$ and $v$ respectively, as well as their approximate PPR values.
Then HOE considers each vertex pair $[u',v']\in H(u)\times H(v)$ (i.e., $[u',v']$ is infected by $[u,v]$).
If either $u'$ or $v'$ is already matched, the algorithm does nothing and continues to consider the next infected pair (line 4-6).
A score $\mathtt{sc}$ of $[u',v']$ w.r.t. $[u,v]$ is computed based on approximate PPR values $\pi^o(u,u')$ and $\pi^o(v,v')$
(line 7, the parameter $\sigma$ will be explained shortly).
If we have $[u',v']\in C$ already, HOE simply updates the matching score of $[u',v']$ by adding $\mathtt{sc}$ to it (line 8-9).
If $[u',v']\not\in C$, then HOE computes the basic matching score of $[u',v']$ w.r.t. all initial seed pairs for the first time,
and let the matching score of $[u',v']$ be the sum of its basic matching score and $\mathtt{sc}$; lastly, HOE adds $[u',v']$ into the candidate set $C$ (line 11-12).

We also provide a more efficient variant, namely the Neighbor-Expansion (NE).
The NE method is more reserved in both candidate expansion and impact expansion.
Given an unused pair of matched vertices $[u,v]\in M'$, NE conducts a $1$-step Forward-Push on $u$ and $v$ respectively, i.e.,
it assigns each $u'\in N(u)$ with $\pi^o(u,u')=\frac{\SP \PP}{|N(u)|}$, and each $v'\in N(v)$ with $\pi^o(v,v')=\frac{\SP \PP}{|N(v)|}$.
Then NE expands the candidate sets by adding each pair $[u',v']\in N(u)\times N(v)$ to $C$ for $[u',v']$ not in $C$ yet; the score of $[u',v']$ is set/updated accordingly.


Given an unused matched pair $[u,v]$, if we use NE for candidate expansion, then, for each neighboring pair $[u',v']\in N(u)\times N(v)$, the matching score $\mathtt{Score}(u', v')$ is incremented by
$$\frac{\min{(\pi^o(u, u'),\pi^o(v, v'))}}{\max{(\pi^o(u, u'),\pi^o(v, v'))}}=
\min\{\frac{|N(u)|}{|N(v)|},\frac{|N(v)|}{|N(u)|}\}.$$
Clearly, using the NE strategy, PPRGM consider a pair $[u,v]$ with smaller relative degree difference as a stronger evidence for later matching. This is interesting as it is in the same spirit as a heuristic used in several previous graph matching algorithms, which prefers pairs with smaller degree difference, e.g.~\cite{kazemi2015growing,yacsar2018iterative}. Thus our PPR-based expansion mechanism (including NE and HOE) and score function essentially applying such intuitive heuristics implicitly, which explains its superior performance in some sense.

\vspace{0.1cm}
\noindent
\textbf{Comparison of HOE and NE.}
The NE method is more efficient than HOE, because it only adds order-1 neighboring pairs and consequently examines less number of candidate pairs overall. HOE could potentially identify more correct matches than NE, i.e., higher recall, since more pairs are added to the candidate set in each expansion.
Detailed experimental comparison between HOE and NE, as well as the effect of various $r'_{max}$ will be presented in the experiment section.

\subsection{A Robust Matching Score Function}
In practice, the score function~\eqref{ep_score1} is not robust against noise and may suffer from numerical issues. Suppose $p_1, p_2$ are two small PPR values, then small additive errors in computing $p_1$ and/or $p_2$ will change the score $\frac{\min{(p_1,p_2)}}{\max{(p_1,p_2)}}$ significantly, which makes the score very unstable. Such errors can be caused by noise in data, limited numerical precision or by approximation errors from PPR computations (i.e., the Forward-Push method).
Therefore, we introduce a smoothing parameter $\sigma$ to improve its robustness. Formally, we define the matching score of $[u,v]$ w.r.t. a seed pair (or a matched pair) $[s_k, s'_k]$ as
\begin{equation}\label{ep_score3}
\mathtt{Score_k}(u,v) = \frac{\min{(\pi(s_k, u),\pi(s'_k, v))}}{\max{(\pi(s_k, u),\pi(s'_k, v))}+\sigma}.
\end{equation}

The extra smoothing term $\sigma$ in the denominator reduces the impact from small PPR values on the matching scores, and thus makes the score much more stable. This might also improve the matching precision since the vertices with small PPR values are often far away from the corresponding seed, and thus two similar but small PPR values may not be a good evidence for matching.


We empirically set $\sigma$ to be $10$ times the residue threshold $r_{max}$ used in Forward-Push.

\subsection{Analysis} \label{33analysis}

\vspace{0.1cm}
\noindent
\textbf{Time complexity analysis.}
We assume the size of both input graphs are of the same order, i.e., both contain $O(|V|)$ vertices and $O(|E|)$ edges for simplicity. We also assume $|S|\ll |V|$ (only a small number of seeds are provided) and $r_{max}\ll r'_{max}$.
The time complexity of our algorithms are summarized as follows.

\begin{lemma}\label{lemma:time}
Given graphs $G_1$ and $G_2$ to be matched, each with $O(|V|)$ vertices and $O(|E|)$ edges, and a set $S$ of seed pairs, the running of PPRGM is $O(\min (\frac{|S|}{r^2_{\max}}, |S||C|)+\frac{|V|}{r'^2_{\max}})$,  which is bounded by $O(\min (\frac{|S|}{r^2_{\max}}, \frac{|S||V|}{r'^2_{\max}}))$ in the worst case,
where $r_{max}$ and $r'_{max}$ are the residue thresholds used in Forward-Push on seed vertices and matched vertices respectively.
\end{lemma}

\begin{proof}
The running time of PPRGM is dominated by Forward-Push and the computation of matching scores, which will be analyzed separately.

The time complexity of Forward-Push  (Algorithm~\ref{alg:FP}) is determined by the residue threshold $r_{\max}$. 
The authors of~\cite{andersen2006local} prove that Forward-Push takes $O(1/r_{\max})$ time.
Our algorithm employs a smaller residue thresholds $r_{\max}$ in Forward-Push for seed vertices, which takes $O(|S|/r_{\max})$ time in total.  Additionally, PPRGM conducts at most one Forward-Push for all vertices with a larger $r'_{\max}$, thus the time needed is $O({|V|}/{r'_{\max}})$ in total. Therefore, the total time for Forward-Push is $O(|S|/r_{\max} + {|V|}/{r'_{\max}})$.
Note if each matched vertex is only allowed to expand for exactly one step (i.e., the NE method), then all Forward-Push processes take $O(\frac{|S|}{r_{max}}+|E|)$ time. 

 A single Forward-Push from a vertex with a residue threshold $r'_{max}$ assigns at most $O(\frac{1}{r'_{max}})$ vertices with a positive PPR value~\cite{andersen2006local}. Therefore, the time to compute matching scores w.r.t. the PPR values assigned by a matched pair $[u,v]$ is $O(\frac{1}{r'^2_{max}})$.
Regarding a maximum number of $|V|$ matched pairs, the cost is $O(\frac{|V|}{r'^2_{max}})$.

Next, we analyze the time for computing matching scores w.r.t. initial seeds.
There are at most $O(\frac{1}{r_{max}})$ vertices with a positive PPR value w.r.t. any seed vertex. Therefore, the time complexity of matching score computation is bounded by $O(\frac{1}{r^2_{max}})$ w.r.t. each seed pair and totally $O(\frac{|S|}{r^2_{max}})$ for all $|S|$ seed pairs. On the other hand, one can easily check this time is also bounded by $O(|S|\cdot |C|)$, where $|C|$ is the total number of candidate pairs. Consequently, the time for computing matching scores w.r.t. initial seeds is bounded by $O(\min(|S||C|,\frac{|S|}{r^2_{max}}))$. Typically, we have $|C|\le \frac{1}{r^2_{max}}$ in our experiments.
We claim $|C|$ is bounded by $O(\frac{|V|}{r'^2_{max}})$.
Because the number of matched pairs is bounded by $O(|V|)$ and a candidate expansion on a matched pair leads to $O(\frac{1}{r'^2_{max}})$ more candidate pairs (the set of $H(u)\times H(v)$).

In conclusion, the overall time complexity of our algorithm is dominated by the cost of matching score computation, which is
$O(\min (\frac{|S|}{r^2_{\max}}, |S||C|)+\frac{|V|}{r'^2_{\max}})$. In the worst case, this is bounded by $O(\min (\frac{|S|}{r^2_{\max}}, \frac{|S||V|}{r'^2_{\max}}))$
\end{proof}

In the our experiments,  we always set $r_{\max} \approx \frac{2\cdot |S|}{|V|}$.
Then there are totally $\frac{|S|}{r_{\max}}=O(|V|)$ non-zero PPR values from all seed pairs, and thus, on average, each vertex gets only $O(1)$ PPR values from seed nodes. Employing hash tables, the matching score computation for a candidate pair $[u,v]$, where vertex $u$ has $m$ positive PPR values and vertex $v$ has $n$ positive PPR values, takes $O(\min(m,n))$ time. Therefore, the computation time of the matching scores for a typical candidate w.r.t. all seed pairs is $O(1)$ in the average sense, and takes $O(|C|)$ time in total for all candidate pairs.\footnote{One can construct a bad instance, which still has worst case time $O(|S||C|)$. }

As discussed above, in the PPRGM framework, an upper bound for $r'_{max}$ should be $O(|V|/|E|)$, corresponds to the NE method; the overall time complexity increases as smaller $r¡¯_{max}$ is used, which corresponds to larger range of expansion and more candidate pairs.


\begin{figure}[!tb]
\begin{centering}
\includegraphics [width=3.0in]{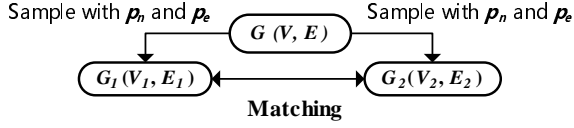}
\vspace*{-2mm}
\caption{$G(n,p;p_n,p_e)$ Random Graph Model. }
\label{randomgraph}
\vspace*{-4mm}
\end{centering}
\end{figure}

\vspace{0.1cm}
\noindent
\textbf{Discrimination power of PPR.}
Let $[s,s']\in S$ be a seed pair and let $u\in V_1$ and $u'\in V_2$ be a valid pair. Recall the definition of $q_{su}^{(t)}$, which is the probability that a non-decaying random walk starting from $s$ reach $u$ after $t$ steps. Let $\chi_u^{L}$ be the $L$-dimensional vector $ \{q_{su}^{(1)},\cdots,q_{su}^{(L)}\}$. The vector $\chi_u^{L}$ provides rich structural information of $u$ w.r.t. to the source vertex $s$, which is often used as a characteristic vector in ranking tasks (see e.g.,~\cite{kloumann2017block}). In particular, when the two graphs $G_1$ and $G_2$ to be matched are isomorphic to each other, then $\chi_u^{L}=\chi_{u'}^{L}$; 
conversely, for any vertex $v$ s.t. $[u,v]$ is not a valid matching, one would expect $\chi_u^{L}\neq \chi_{v}^{L}$, i.e., $q_{su}^{(t)}\neq q_{sv}^{(t)}$ for at least one $t\in[L]$ (for large enough $L$).
However, directly comparing the characteristic vectors of $u$ and $v$ to decide whether they are a true match is inefficient. Hence it is crucial to devise a \emph{discriminant function} $f:\real^{L}\rightarrow \real$ such that, whenever $\chi_u^{L}\neq \chi_v^{L}$, it holds $f(\chi_u^{L})\neq f(\chi_v^{L})$.  The truncated PPR value is such a discriminant function, which is the weighted combination of the entries $\pi^L(s,u)  = \alpha \cdot \sum_{t=0}^L \PP^t \cdot q_{su}^{(t)}$.
\begin{lemma}
	Suppose $\chi_u^{L}\neq \chi_v^{L}$, and $\alpha$ is randomly chosen from the interval $[0,1]$, then
	$\Pr[\pi^L(s,u) = \pi^L(s,v)]=0.$
\end{lemma}
\begin{proof}
	We define the polynomial $p_L(\alpha)$ as follows.
	$$p_L(\alpha)=\frac{\pi^L(s,u) - \pi^L(s,v)}{\alpha} = \sum_{t=0}^L \PP^t \cdot (q_{su}^{(t)} - q_{sv}^{(t)}).$$
	By assumption $\chi_u^{L}\neq \chi_v^{L}$, it follows that $p_L(\alpha)$ is a non-zero polynomial with degree at most $L$.
	Consequently, the probability that $\alpha$ is a root of $p_L$ is $0$, since $\alpha$ is a uniform random number in $[0,1]$ and the number of roots of $p_L(\alpha)$ is finite (at most $L$).
\end{proof}
One could see, in the noiseless setting where $G_1$ is isomorphic to $G_2$,  larger values of $L$ provide better theoretical discrimination; this is equivalent to computing PPR values with higher precision.

But in practice, $G_1$ and $G_2$ are often considered as two noisy versions of a ``ground truth" graph $G$. Thus the PPR value of $u$ in $G_1$, namely $\pi_{G_1}(s,u)$, is merely an estimate of $\pi_G(s,u)$. In this case, if the random walk in $G_1$ takes more steps, then it would accumulates more noise, which makes the estimates less reliable. Therefore, the choice of $L$ can be considered as a \emph{reliability versus discrimination} trade-off. In PPRGM, we implicitly control the value of $L$ by adjusting the parameter $r_{\max}$ in Forward-Push. While our original motivation to use smaller $L$ is to save computation time, doing this will also improve the precision sometimes as evidenced by our experiments.

\vspace{0.1cm}
\noindent
\textbf{Efficacy of the postponing strategy.}
Let us first introduce the $G(n,p;p_n,p_e)$ random graph model that has been used in~\cite{pedarsani2011privacy, kazemi2015growing} for analyzing percolation-based graph matching algorithms.
The $G(n,p;p_n,p_e)$ random graph model (illustrated in Figure~\ref{randomgraph}) generates two correlated graphs as follows:
(1) A graph $G(V,E)$ is generated from Erd\"os-Renyi random graph model, i.e., in a graph of $n$ vertices, each of ${n}\choose{2}$ possible edges occurs with probability $0<p<1$;
(2) $V_1$ and $V_2$ are two independent sample sets of $V$, in which each vertex $v\in V$ is sampled with probability $p_n$;
(3) Each edge $(u,v)\in E$ s.t. $u,v\in V_1$ is included in $E_1$ with probability $p_e$, and $E_2$ is generated in the same way.
We assume that the average degree $np$ is constant, 
thus $G(V,E)$ has a unique giant component containing a positive fraction of the vertices with high probability. The task of graph matching is to identify the correct matches between $V_1$ and $V_2$.

To prove the effectiveness of postponing strategy, we consider a special case here:
(1) the order-1 (neighboring) information is considered only, and each seed match adds $1$ mark to its neighboring pair, i.e., the matching score of pair $[u,v]$ is the number of matched pairs in $N(u,v)$;
(2) we fix $\gamma=1+\epsilon_1$ and $\beta=\epsilon_2$ where $\epsilon_1$ and $\epsilon_2$ are small positive constants. That is to say, a pair $[u,v]$ can be matched if $\mathtt{Score}(u,v)\ge 2$,
and $\mathtt{Score}(u,v)\ge\mathtt{Score}(u',v')+1$ for any $[u',v']\in \{(\{u\}\times V_2)\cup (V_1\times \{v\})-[u,v]\}$.

Denote $n_c$ the number of correctly matched pairs in set $S$ of seed matches, and $n_w$ the number of wrongly matched pairs in $S$. We give the following lemma.
\vspace*{-1mm}
\begin{lemma}\label{lemma1}
In the $G(n,p;p_n,p_e)$ random graph model, assume the above percolation method and matching criteria are adopted, in each iteration, the postponing strategy decreases the probability of a wrong pair being matched
from $(n_c+n_w)^2 p^4p_e^4$ to $(n_c+n_w)^2 p^4 p_e^4 - n_c^2 p^4 p_e^4 (2p_e^4-p_e^8+ 2 p_n^2 p_e^4 - 2p_n^4 p_e^4)$, and decreases the probability of a correct pair being matched
from $n_c^2 p^2 p_e^4$ to $n_c^2 p^2 p_e^4 (1-np^2)$.
\end{lemma}
\vspace*{-1mm}
The proof of Lemma~\ref{lemma1} is included in the full version due to space constraints. Obviously, the postponing strategy greatly decreases the number of pairs that are wrongly matched.
Assume $p_n=1$ (say $V_{1,2}=V$) and $n_w=0$ (no wrong matches in the seed set),
then the probability of a wrong pair $[v,u]$ being matched in the next iteration is $p^4p_e^4 n_c^2(1-p_e^4)^2$, which is decreased nearly by $10$ times when the edge sampling probability $p_e=0.9$, compared to a probability of $p^4 p_e^4 n_c^2$ without the postponing strategy. In two isomorphic graphs, with $p_n=1$ and $p_e=1$, then no wrong pairs can be matched if all initial seeds are correct. This will also be verified in the experiments as our algorithm achieves a precision of $1$ in matching two isomorphic graphs.
In the case when $np\ll \sqrt{n}$ or $\frac{1}{p^2}\gg n$ (the random graph is relatively sparse, which is the most interesting case in practice), the probability of matching a correct pair with postponing, i.e., $n_c^2 p^2 p_e^4 (1-np^2)$, is very close to $n_c^2 p^2 p_e^4$, i.e., the probability of matching a correct pair without postponing. By replacing the matching probability (for correct pairs and wrong pairs) in Theorem 1 in~\cite{kazemi2015growing} with the new probabilities derived above, we can conclude that the postponing strategy increases the matching precision greatly while not affecting the overall recall by much.

\vspace{0.1cm}
\noindent
\textbf{Property of PPR-based signature vector.}
We will show that the $signature$ vector $R(u)$ of vertex $u$ illustrates $u$'s region information and local structural information.
A PPR value $\pi(s, u)$ of $u$ w.r.t. a seed $s$ can be decomposed into two parts:
(1) the probability that a decaying random walk terminates at vertex $u$ the first time it arrives at $u$, denoted by \textbf{$\pi_1(s, u)$};
(2) the probability that a decaying random walk terminates at vertex $u$ after arriving at vertex $u$ for more than once, denoted by \textbf{$\pi_2(s, u)$}.

Following the analysis in Section~\ref{22PPR}, we use $LS(u) = [\sum_{t=1}^{\infty}\PP^{t}\cdot P^{t}\cdot e_u]_u$ to denote the summed probabilities that a decaying random walk starting from $u$ reaches $u$ after step $1,2,3 \cdots$, then
$$\pi_2(s, u) = \pi_1(s, u) \times LS(u).$$
Then, the PPR value $\pi(s, u)$ can be presented by
$$\pi(s, u) = \pi_1(s, u) +\pi_2(s, u)  = \pi_1(s, u) \cdot(1+LS(u)).$$ 


Recall that $\frac{\pi_1(s, u)}{\SP}$ represents the summed probability that a decaying random walk from $s$ hits vertex $u$ for the first time at iteration $0, 1, 2\cdots$.
With a sufficient number of initial seeds in $S$, the region information of vertex $u$ can be represented by all $\frac{\pi_1(s_k, u)}{\SP}$ for $1\le k \le |S|$.
More specifically, given a number of correct seed matches on two graphs, $u\in V_1$ and $v\in V_2$ are likely on the aligned region over two graphs if and only if $\frac{\pi_1(s_k, u)}{\SP}$ and $\frac{\pi_1(s'_k, v)}{\SP}$ are close with respect to a certain number of seeds $[s_k, s'_k]\in S$.

On the other hand, 
note that $LS(u) = [\sum_{t=1}^{\infty}\PP^{t}\cdot P^{t}\cdot e_u]_u$ is independent of the source vertex of random walk; $LS(u)$ only depends on vertex $u$ and graph $G$.
We say $LS(u)$ demonstrates the graph structural information around vertex $u$; the closer $LS(u)$ and $LS(v)$ are, the local structures around $u$ and $v$ are more likely to be similar.
For instance, consider the probability that a random walk starts from $u$ that arrives at $u$ within $2$ steps, denoted by $LS^{(2)}(u)$.
Assume there is no self-loop in graph, the random walk from $u$ cannot return $u$ after one step, then
\begin{equation}
\begin{split}
LS^{(2)}(u) & = [\PP^{2}\cdot P^{2}\cdot e_u]_u  \\
            & = \PP^{2}\cdot \frac{1}{|N(u)|} \sum_{(u,v)\in E}\frac{1}{|N(v)|}
\end{split}
\end{equation}
Thus $LS^{(2)}(u)$ is the average reciprocal of degree of all neighbors of vertex $u$, containing the order-2 structural information of $u$.
Accordingly, $LS^{(t)}(u)$ with $t>2$ contains higher order structural information of $u$, the impact of which, however, is decayed by the factor $\PP^t$ over iterations.

\section{Experimental results}


\begin{table}[!tb]
\begin{center}
\vspace*{-2mm}\caption{\label{table:datasets1}
\bf{Datasets of Sampling Construction.}}
\vspace{-2mm}
\renewcommand{\arraystretch}{1}
\small
\resizebox{\columnwidth-0.2in}{!}{%
\begin{tabular}{l|cccc}
\noalign{\hrule height 0.7pt}
\small
\textbf{Datasets} & $|V|$ & $|E|$ & Triangles & $avg_D$ \\
\hline\hline

\textbf{Twitter}~\cite{leskovec2012learning} & $81.3$K & $1.7$M & $13.1$M & $3.7$\\

\textbf{Dblp}~\cite{yang2015defining} & $317.1$K & $1.1$M  & $2.2$M & $6.2$\\

\textbf{Amazon}~\cite{leskovec2009community} & $334.9$K & $925.9$K  & $667.1$K & $11.7$\\

\textbf{Youtube}~\cite{yang2015defining} & $1.1$M & $3.0$M  & $3.0$M & $4.2$\\

\textbf{WikiTalk}~\cite{leskovec2010predicting} & $2.4$M & $5.0$M  & $9.2$M &  $3.6$\\ 
\noalign{\hrule height 0.7pt}

\end{tabular}
}
\end{center}
\end{table}

\begin{table}[!tb]
\begin{center}
\vspace*{-2mm}\caption{\label{table:datasets2}
\bf{Snapshots of Superuser and AskUbuntu.}}
\vspace{-2mm}
\hspace*{-2mm}
\renewcommand{\arraystretch}{1.1}
\resizebox{\columnwidth}{!}{%
\begin{tabular}{l|cccc}
\noalign{\hrule height 0.7pt}
\textbf{Period} & $|V|$ & $|E|$ & $|V_1\cap V_2|$ & $|E_1\cap E_2|$ \\
\hline\hline
\multicolumn{4}{l}{\textbf{Superuser}~\cite{paranjape2017motifs}, \ $avg_D\in [3.83,3.85]$  }\\
\hline
\textbf{P0:}Jan,2014-Jan,2016 & $107$K & $241$K & $/$ & $/$\\

\textbf{P1:}Oct,2013-Oct,2015 & $105$K & $239$K  & $96$K & $209$K \\

\textbf{P2:}Jul,2013-Jul,2015 & $103$K & $238$K  & $84$K & $175$K \\

\textbf{P3:}Apr,2013-Apr,2015 & $100$K & $234$K  & $74$K &  $146$K \\
\hline
\multicolumn{4}{l}{\textbf{AskUbuntu}~\cite{paranjape2017motifs},  \ $avg_D\in [3.73,3.76]$ }\\
\hline
\textbf{P0:}Jan,2013-Jan,2016 & $120$K & $294$K & $/$ & $/$\\

\textbf{P1:}Sep,2012-Sep,2015 & $117$K & $290$K  & $108$K & $260$K\\

\textbf{P2:}May,2012-May,2015 & $114$K & $288$K  & $96$K & $224$K\\

\textbf{P3:}Jan,2012-Jan,2015 & $108$K & $278$K  & $85$K &  $191$K\\
\noalign{\hrule height 0.7pt}
\end{tabular}
}
\vspace*{-5mm}
\end{center}
\end{table}


%
%
%

In the experiments, we compare the PPRGM algorithms with the state-of-the-art methods on datasets with various characteristics. All algorithms are evaluated with respect to accuracy, efficiency and robustness, and it is observed that PPRGM outperforms the state-of-the-art on all aspects. We also study the performance of PPRGM with varying parameters. 

\subsection{Experiment Setting}\label{51setup}

\vspace{0.1cm}
\noindent
\textbf{Datasets.} We use multiple publicly available datasets. Two input graphs to be matched are constructed in two ways: (1) vertex/edge sampling from a real graph; (2) snapshots of a real temporal network with respect to different time windows. 

The way of constructing correlated graphs by vertex/edge sampling is similar to the method of generating $G(n,p;p_n,p_e)$ random graphs (see Figure~\ref{randomgraph}): given a real graph $G(V,E)$ (in stead of a Erd\"os-Renyi random graph), two subgraphs $G_1(V_1,E_1)$ and $G_2(V_2,E_2)$ are independently sampled from $G(V,E)$ with the vertex sampling probability $p_n$ and edge sampling probability $p_e$.
Two subgraphs are isomorphic when $p_n=1$ and $p_e=1$.
The datasets for the vertex/edge sampling approach are summarized in Table~\ref{table:datasets1}, and their detailed descriptions can be found in SNAP\footnote{http://snap.stanford.edu/data}.
Note that the average distance $avg_D$ for each network is presented.
Given a graph $G(V,E)$, we denote $P$ the set of all $(u,v)\in V\times V$ s.t. $u$ and $v$ are connected. Let $d(u,v)$ be the shortest path distance from $u$ to $v$. The average distance of $G$ is given by $avg_D= \frac{1}{|P|}\sum_{(u,v)\in P}{d(u,v)}$.





We also generate the input graphs from two temporal networks Superuser and AskUbuntu~\cite{paranjape2017motifs}.
Their detailed descriptions can also be found in SNAP. Each snapshot of a temporal network is generated by aggregating the edges within a given time period (time window).
We construct four snapshots $P_i$ for $0\le i \le 3$ for each dataset, and the summaries of them are presented in Table~\ref{table:datasets2}. 
For each dataset, the snapshot P0 shares different lengths of overlapping period with other three snapshots; the number of vertices and edges shared by two snapshots decrease as the length of their overlapping period decreases.



In the experiments, we use the default vertex sampling probability $p_n=1$ and the edge sampling probability $p_e=0.8$.
For time-varying datasets, we by default consider the matching between snapshots P0 and P1.

\noindent
\textbf{Competing algorithms.}
We evaluate the performance of PPRGM {High-Order Expansion} and {Neighbor Expansion}, referred to as HOE and NE respectively. The proposed methods are compared with four previous algorithms, which are listed below:
\vspace{-0.2cm}
\begin{itemize}
  \item EWS: The state-of-the-art PGM algorithm proposed in~\cite{kazemi2015growing} (its basic idea is introduced in Section~\ref{sec1}). It has been shown in \cite{kazemi2015growing} that EWS significantly outperforms other percolation based algorithms such as User-Matching~\cite{korula2014efficient} and PercolateMatched~\cite{yartseva2013performance}.
  \vspace{-2mm}
  \item FRUI: A variant of PGM proposed in~\cite{zhou2016cross}, which uses additional criteria to break the tie when multiple pairs have the same number of matched neighboring pairs (see Section~\ref{sec1}).
  \vspace{-2mm}
  \item AE: The anchor-expansion algorithm from~\cite{zhu2013high}. The similarity of each vertex pair is measured by comparing the two degree sequences extracted from their 2-neighborhood subgraphs; the algorithm iteratively matches the pair with the highest similarity score.
  \vspace{-2mm}
  \item GSANA: An algorithm proposed in~\cite{yacsar2018iterative} that finds candidates by an embedding strategy. Vertices of the two graphs are first embedded onto the same 2-D space based on their shortest path distances to some selected seeds; a pair $[u,v]$ is measured (i.e., considered as a candidate pair) if $u$ and $v$ are close in the 2-D space.
      It has been shown in~\cite{yacsar2018iterative} that GSANA outperforms IsoRank~\cite{singh2008global}, Klau~\cite{klau2009new}, and NetAlign~\cite{bayati2009algorithms}.
\end{itemize}
  \vspace{-2mm}
The codes of EWS and AE are provided the original authors; we thank them for kindly sharing their codes. All algorithms are implemented in C++.



\noindent
\textbf{Metrics.}
To assess the quality of matching algorithms, we use the same definition of \emph{precision} and \emph{recall} as in \cite{kazemi2015growing}: $precision=\frac{n_c}{n_c+n_w}$ where $n_c$ is the number of correct matches and $n_w$ is the number of wrong matches; $recall=\frac{n_c}{n_{ident}}$ where $n_{ident}$ is the number of vertices that are present in both graphs with degrees at least two.
We also evaluate the F1-score, 
the harmonic mean of precision and recall; formally $F1$-$score = 2\cdot \frac{precision\times recall}{precision+recall}$.

\vspace{0.1cm}
\noindent
\textbf{Setup.}
The number of seeds required is always a great concern in graph matching.
We set the number of seeds to $20$ by default for all datasets. 
Each correct seed $[u,u]$ is uniformly sampled from $V'=V_1\cap V_2$. We also consider the case when the initial seed set contains some wrong pairs and test the performance of the algorithms against such noisy seeds. In our experiments, a wrong seed $[u,v]$ is randomly sampled from $V_1\times V_2-V'\times V'$. 
We empirically set the default $\SP$ to $0.3$ following the discussion in Section~\ref{33analysis}.
The $r_{max}$ used in Forward-Push on initial seeds is set s.t. $|S|/r_{max}$ approximates $2\cdot max(|V_1|,|V_2|)$. 
The residue threshold $r'_{max}$ used in Forward-Push on other matched vertices (in HOE) is set to $10^{-3}$ by default; the effect of choosing different $r'_{max}$ on different datasets are also provided.

We report the average performance of $10$ tests in all experiments.
In each test, we first generate two graphs and a set of seed pairs according to the methods described above,
which are then used as the input for all algorithms.
We run experiments on a Linux machine with Intel Xeon E5-2698 v4 cloked at 2.2GHz.

\subsection{Algorithm Comparison}\label{52algorithmcomparison}

\begin{table*}[!tb]
  \centering
  \small
  \vspace*{-2mm}
  \caption{\bf{Precision and Recall on Datasets}}
  \vspace{-2mm}
  \label{table:pr_results1}
  \renewcommand{\arraystretch}{1.02}
  \resizebox{\columnwidth*2+0.2in}{!}{%
  \begin{tabular}{l|c|ccc|ccc|ccc|ccc|ccc|ccc}
    \noalign{\hrule height 1pt}
    \multirow{2}{*}{Datasets} & \multirow{2}{*}{Alg.} &
        \multicolumn{3}{c|}{$p_n$=$1.0$, $p_e$=$1.0$} & \multicolumn{3}{c|}{$p_n$=$1.0$, $p_e$=$0.9$} & \multicolumn{3}{c|}{$p_n$=$1.0$, $p_e$=$0.8$} &
        \multicolumn{3}{c|}{$p_n$=$0.9$, $p_e$=$0.9$} & \multicolumn{3}{c|}{$p_n$=$0.9$, $p_e$=$0.8$} & \multicolumn{3}{c|}{$p_n$=$0.8$, $p_e$=$0.9$}  \\
        \cline{3-20}
        & & \Recall & \Precison & \FSCORE & \Recall & \Precison & \FSCORE & \Recall & \Precison & \FSCORE & \Recall & \Precison & \FSCORE & \Recall & \Precison & \FSCORE & \Recall & \Precison & \FSCORE \\
    \hline
    \hline
    \multirow{6}{*}{\textbf{Twitter}}
        & HOE & \textbf{0.980} & \textbf{1.000} & \textbf{0.990} & \textbf{0.931} & 0.962 & \textbf{0.946} & \textbf{0.854} & 0.889 & \textbf{0.871} & \textbf{0.870} & 0.842 & \textbf{0.856} & \textbf{0.762} & 0.745 & \textbf{0.753} & \textbf{0.799} & 0.716 & \textbf{0.755} \\
        & NE & 0.965 & \textbf{1.000} & 0.982 & 0.906 & \textbf{0.965} & 0.935 & 0.823 & \textbf{0.896} & 0.858 & 0.840 & \textbf{0.853} & 0.846 & 0.733 & \textbf{0.766} & 0.749 & 0.766 & \textbf{0.736} & 0.751 \\
        & EWS & 0.966 & 0.968 & 0.967 & 0.900 & 0.912 & 0.906 & 0.763 & 0.793 & 0.778 & 0.815 & 0.772 & 0.793 & 0.526 & 0.539 & 0.532 & 0.630 & 0.568 & 0.598 \\
        & GSANA & 0.901 & 0.915 & 0.908 & 0.144 & 0.271 & 0.188 & 0.118 & 0.236 & 0.157 & 0.118 & 0.223 & 0.154 & 0.096 & 0.207 & 0.131 & 0.109 & 0.187 & 0.138  \\
        & AE & 0.911 & 0.912 & 0.912 & 0.338 & 0.352 & 0.345 & 0.145 & 0.157 & 0.150 & 0.123 & 0.122 & 0.123 & 0.057 & 0.058 & 0.058 & 0.052 & 0.048 & 0.050 \\
        & FRUI & 0.873 & 0.890 & 0.882 & 0.276 & 0.323 & 0.298 & 0.306 & 0.340 & 0.322 & 0.354 & 0.361 & 0.358 & 0.032 & 0.036 & 0.034 & 0.133 & 0.128 & 0.130   \\
    \hline
    \multirow{6}{*}{\textbf{Dblp}}
        & HOE & \textbf{0.970} & \textbf{1.000} & \textbf{0.985} & \textbf{0.652} & 0.712 & \textbf{0.681} & \textbf{0.535} & 0.596 & \textbf{0.564} & \textbf{0.542} & 0.585 & \textbf{0.563} & \textbf{0.422} & 0.463 & \textbf{0.441} & \textbf{0.413} & 0.432 & \textbf{0.422}  \\
        & NE & 0.767 & \textbf{1.000} & 0.868 & 0.482 & \textbf{0.737} & 0.583 & 0.371 & \textbf{0.635} & 0.468 & 0.382 & \textbf{0.628} & 0.475 & 0.278 & \textbf{0.501} & 0.357 & 0.274 & \textbf{0.491} & 0.352  \\
        & EWS & 0.618 & 0.679 & 0.648 & 0.491 & 0.589 & 0.536 & 0.358 & 0.478 & 0.410 & 0.390 & 0.476 & 0.429 & 0.264 & 0.355 & 0.303 & 0.279 & 0.337 & 0.305  \\
        & GSANA & 0.875 & 0.934 & 0.903 & 0.018 & 0.166 & 0.032 & 0.010 & 0.128 & 0.019 & 0.010 & 0.115 & 0.018 & 0.009 & 0.111 & 0.016 & 0.007 & 0.098 & 0.014  \\
        & AE & 0.749 & 0.750 & 0.749 & 0.303 & 0.358 & 0.329 & 0.095 & 0.128 & 0.109 & 0.102 & 0.128 & 0.114 & 0.032 & 0.043 & 0.037 & 0.031 & 0.037 & 0.034   \\
        & FRUI & 0.234 & 0.349 & 0.281 & 0.001 & 0.001 & 0.001 & 0.000 & 0.001 & 0.000 & 0.000 & 0.000 & 0.000 & 0.000 & 0.000 & 0.000 & 0.000 & 0.000 & 0.000   \\
    \hline
    \multirow{6}{*}{\textbf{Amazon}}
        & HOE & \textbf{0.981} & \textbf{1.000} & \textbf{0.990} & \textbf{0.716} & \textbf{0.755} & \textbf{0.735} & \textbf{0.508} & \textbf{0.556} & \textbf{0.531} & \textbf{0.538} & \textbf{0.569} & \textbf{0.553} & \textbf{0.348} & \textbf{0.392} & \textbf{0.369} & \textbf{0.334} & \textbf{0.373} & \textbf{0.353}  \\
        & NE & 0.785 & \textbf{1.000} & 0.880 & 0.454 & 0.754 & 0.567 & 0.130 & 0.328 & 0.186 & 0.106 & 0.289 & 0.155 & 0.013 & 0.044 & 0.019 & 0.008 & 0.029 & 0.013 \\
        & EWS & 0.671 & 0.774 & 0.719 & 0.401 & 0.571 & 0.471 & 0.071 & 0.160 & 0.098 & 0.075 & 0.157 & 0.101 & 0.003 & 0.009 & 0.005 & 0.002 & 0.006 & 0.003 \\
        & GSANA & 0.958 & 0.973 & 0.966 & 0.010 & 0.206 & 0.020 & 0.006 & 0.183 & 0.011 & 0.006 & 0.176 & 0.011 & 0.005 & 0.162 & 0.009 & 0.004 & 0.157 & 0.008  \\
        & AE & 0.898 & 0.899 & 0.898 & 0.353 & 0.409 & 0.379 & 0.075 & 0.106 & 0.088 & 0.075 & 0.099 & 0.085 & 0.016 & 0.024 & 0.019 & 0.014 & 0.019 & 0.016 \\
        & FRUI & 0.141 & 0.212 & 0.169 & 0.000 & 0.001 & 0.000 & 0.000 & 0.000 & 0.000 & 0.000 & 0.000 & 0.000 & 0.000 & 0.000 & 0.000 & 0.000 & 0.000 & 0.000  \\
    \hline
    \multirow{3}{*}{\textbf{Youtube}}
        & HOE & \textbf{0.908} & \textbf{1.000} & \textbf{0.952} & \textbf{0.811} & 0.937 & \textbf{0.869} & \textbf{0.714} & 0.866 & \textbf{0.783} & \textbf{0.730} & 0.833 & \textbf{0.778} & \textbf{0.643} & 0.751 & \textbf{0.693} & \textbf{0.652} & 0.721 & \textbf{0.685}   \\
        & NE & 0.868 & \textbf{1.000} & 0.929 & 0.750 & \textbf{0.963} & 0.843 & 0.629 & \textbf{0.938} & 0.753 & 0.660 & \textbf{0.899} & 0.761 & 0.563 & \textbf{0.855} & 0.679 & 0.574 & \textbf{0.832} & 0.679  \\
        & FRUI & 0.266 & 0.441 & 0.332 & 0.101 & 0.189 & 0.132 & 0.000 & 0.000 & 0.000 & 0.040 & 0.068 & 0.050 & 0.000 & 0.000 & 0.000 & 0.000 & 0.000 & 0.000  \\
    \hline
    \multirow{3}{*}{\textbf{WikiTalk}}
        & HOE & \textbf{0.773} & \textbf{1.000} & \textbf{0.872} & 0.730 & 0.968 & 0.832 & \textbf{0.682} & 0.917 & 0.782 & \textbf{0.703} & 0.899 & 0.789 & \textbf{0.650} & 0.870 & \textbf{0.744} & \textbf{0.658} & 0.845 & 0.740 \\
        & NE & 0.772 & \textbf{1.000} & 0.872 & \textbf{0.732} & \textbf{0.975} & \textbf{0.836} & 0.680 & \textbf{0.937} & \textbf{0.788} & 0.702 & \textbf{0.920} & \textbf{0.796} & 0.641 & \textbf{0.876} & 0.740 & 0.657 & \textbf{0.860} & \textbf{0.745}   \\
        & FRUI & 0.000 & 0.000 & 0.000 & 0.000 & 0.001 & 0.000 & 0.000 & 0.000 & 0.000 & 0.000 & 0.001 & 0.000 & 0.000 & 0.000 & 0.000 & 0.000 & 0.000 & 0.000  \\
    \noalign{\hrule height 1pt}
  \end{tabular}
  }
\end{table*}
\vspace{0.1cm}

\begin{table}[!tb]
\small
\begin{center}
\vspace*{-2mm}\caption{\label{table:pr_results2}\bf{Precision and Recall on Datasets}}
\vspace{-2mm}
\renewcommand{\arraystretch}{1}
\resizebox{\columnwidth-0.2in}{!}{%
\begin{tabular}{c|ccc|ccc|ccc}
    \noalign{\hrule height 1pt}
    \multirow{2}{*}{Alg.} & \multicolumn{3}{c|}{P0 vs. P1} & \multicolumn{3}{c|}{P0 vs. P2} & \multicolumn{3}{c}{P0 vs. P3} \\
    \cline{2-10}
    & \Recall & \Precison & \FSCORE  & \Recall & \Precison & \FSCORE  & \Recall & \Precison & \FSCORE  \\
    \hline
    \hline
    \multicolumn{9}{l}{\textbf{Superuser}}\\
    \hline
    HOE & \textbf{0.842} & 0.950 & \textbf{0.893} & \textbf{0.768} & 0.865 & 0.814 & \textbf{0.678} & 0.771 & 0.722  \\
    NE & 0.820 & \textbf{0.977} & 0.892 & 0.737 & \textbf{0.932} & \textbf{0.823} & 0.636 & \textbf{0.888} & \textbf{0.742}  \\
    EWS & 0.574 & 0.643 & 0.607 & 0.414 & 0.527 & 0.464 & 0.139 & 0.227 & 0.172  \\
    GSANA & 0.115 & 0.331 & 0.171 & 0.073 & 0.222 & 0.110 & 0.054 & 0.186 & 0.084  \\
    AE & 0.103 & 0.115 & 0.109 & 0.014 & 0.015 & 0.015 & 0.005 & 0.005 & 0.005  \\
    FRUI & 0.523 & 0.635 & 0.573 & 0.081 & 0.147 & 0.105 & 0.042 & 0.070 & 0.052  \\
    \hline
    \multicolumn{9}{l}{\textbf{AskUbuntu}}\\
    \hline
    HOE & \textbf{0.863} & 0.966 & 0.912 & \textbf{0.820} & 0.911 & 0.863 & \textbf{0.776} & 0.847 & 0.810  \\
    NE & 0.855 & \textbf{0.982} & \textbf{0.914} & 0.809 & \textbf{0.952} & \textbf{0.875} & 0.766 & \textbf{0.912} & \textbf{0.832}   \\
    EWS & 0.609 & 0.673 & 0.639 & 0.452 & 0.570 & 0.504 & 0.353 & 0.466 & 0.402   \\
    GSANA & 0.277 & 0.518 & 0.361 & 0.105 & 0.278 & 0.153 & 0.084 & 0.221 & 0.122  \\
    AE & 0.084 & 0.097 & 0.090 & 0.013 & 0.014 & 0.014 & 0.008 & 0.008 & 0.008   \\
    FRUI & 0.345 & 0.531 & 0.418 & 0.226 & 0.350 & 0.275 & 0.141 & 0.228 & 0.174  \\
    \noalign{\hrule height  1pt}
\end{tabular}
}
\vspace*{-3mm}
\end{center}
\end{table}

\vspace{0.1cm}
\noindent
\textbf{1. Recall and Precision.} Firstly, we evaluate the quality of the output matching.
Table~\ref{table:pr_results1} presents the experimental results (``R'' stands for recall and ``P'' stands for precision) on graphs that are constructed by the vertex/edge sampling approach. We use vertex sampling probability $p_n\in\{1.0, 0.9, 0.8\}$ and edge sampling probability $p_e\in\{0.9,0.8\}$. We additionally run matching algorithms on two isomorphic graphs where no vertex/edge sampling is applied (or equivalently $p_n=1$ and $p_e=1$).
Note that when $p_n=0.8$, only a $0.64$ fraction of all vertices exist in both graphs simultaneously (in expectation).
We specifically set $r'_{max}=10^{-4}$ for Amazon dataset which has very large average distance (use $r'_{max}=10^{-3}$ by default for other datasets), which will be discussed in Section~\ref{53performancestudy}.
Under our machine setting, the EWS, GSANA, and AE algorithms are too slow to finish the experiments on Youtube and WikiTalk dataset (within 6 hours).
In Table~\ref{table:pr_results2}, we show the results on the two temporal datasets Superuser and AskUbuntu. We make the following observations.

(1) HOE and NE have overall better performance than the start-of-the-art methods (among these algorithms EWS has the best overall recalls and precisions).
HOE always has significantly higher recalls and higher precisions than the all previous algorithms;
NE achieves very high precisions (usually even higher than HOE), while it also has competitive recalls. Regarding F1-scores, both HOE and NE outperform all other competing algorithms by a noticeable margin, while HOE has the best performance most of the time and could be significantly better than all other algorithms sometimes (e.g., on the Amazon dataset).
In the cases where two isomorphic graphs are considered ($p_n$=$1$ and $p_e$=$1$), HOE and NE always have $100\%$ precision, while other algorithms still output lots of wrong match pairs. Under this setting, the recalls of HOE are also close to $1$ for most datasets.

(2) HOE has higher recalls than NE on most datasets, while NE often has higher precisions than HOE.
Because PPRGM greedily matches the candidate pairs with high credibility first; it becomes harder to choose correct matches among candidate pairs with relatively lower credibilities.
HOE is more aggressive in matching pairs than NE, which results in higher recall at the cost of slightly lower precision. However, if we stop the matching process of HOE earlier, then at the stage when roughly the same number of correct pairs are matched, HOE often has higher precisions than NE.
\begin{figure*}[!tb]
    \centering
    \begin{minipage}[c]{0.74\textwidth}%
    \captionsetup[subfigure]{oneside,margin={10mm,0cm}}
    \hspace{-6mm}
    \begin{subfigure}{1\textwidth}
        \includegraphics[width=\linewidth]{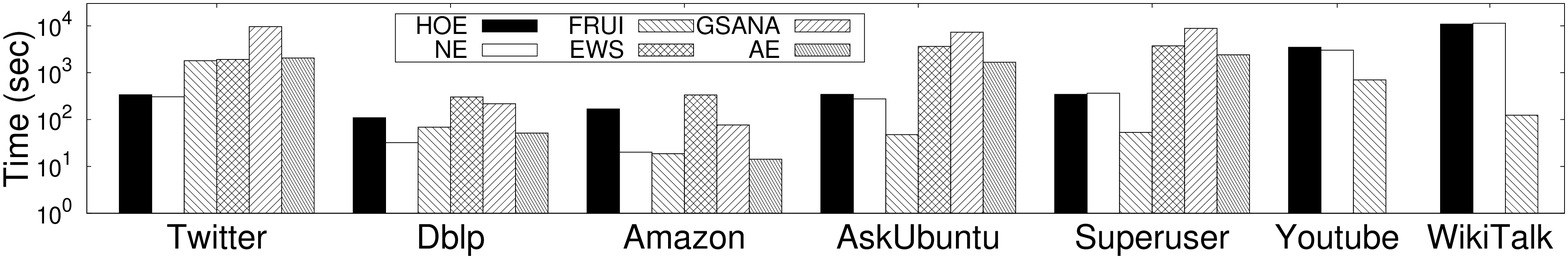}
        \vspace{-4mm}
        \label{fig:time}
        \caption{\label{expfig:time}{Matching time on different datasets.}}

    \end{subfigure}
    \captionsetup[subfigure]{oneside,margin={6mm,0cm}}
    \hspace{3mm}
    \hspace{-40mm}
    \begin{subfigure}{1\textwidth}
        \includegraphics[width=\linewidth]{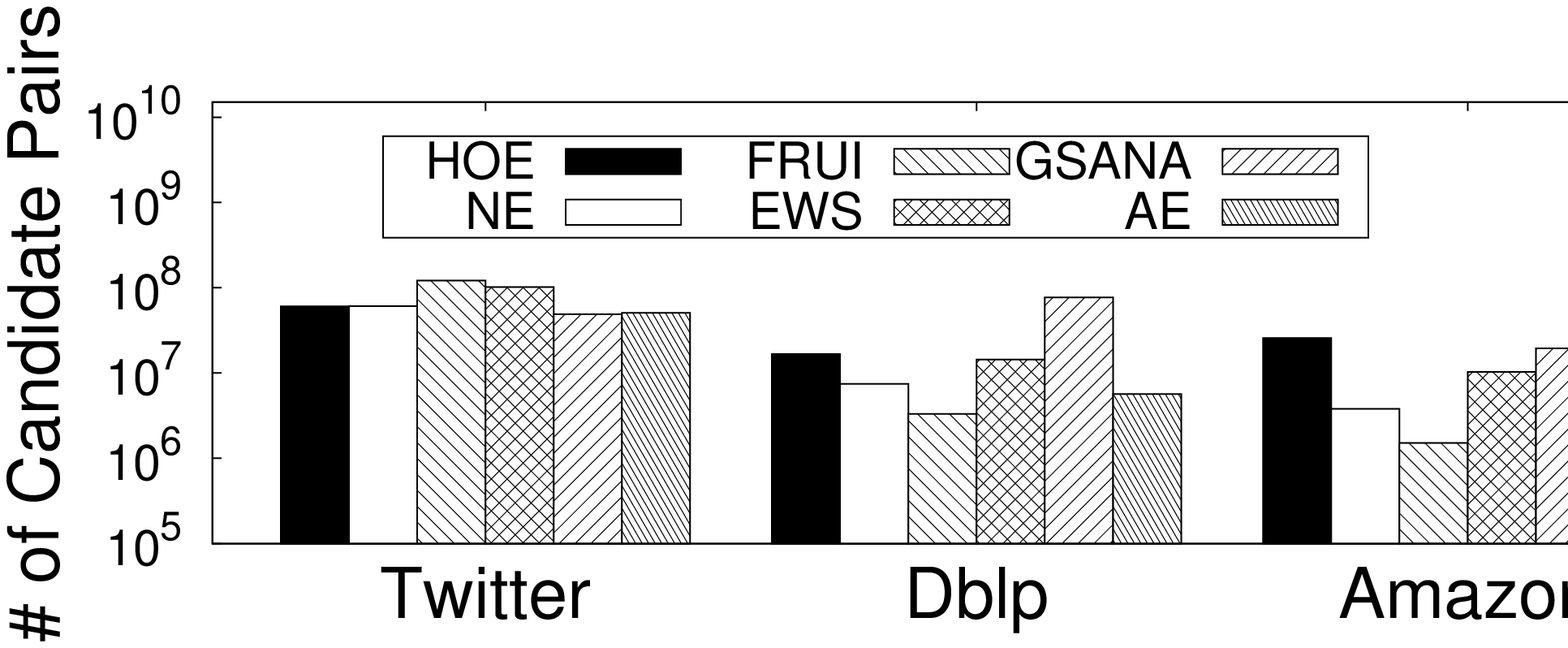}
        \vspace{-4mm}
        \label{fig:candidate}
        \caption{\label{expfig:candidate}{Number of candidate pairs on different datasets.}}

    \end{subfigure}
    \end{minipage}
    \qquad
    \begin{minipage}[c]{0.22\textwidth}%
    \hspace{-2mm}
    \begin{subfigure}{1\textwidth}
        \includegraphics[width=\linewidth]{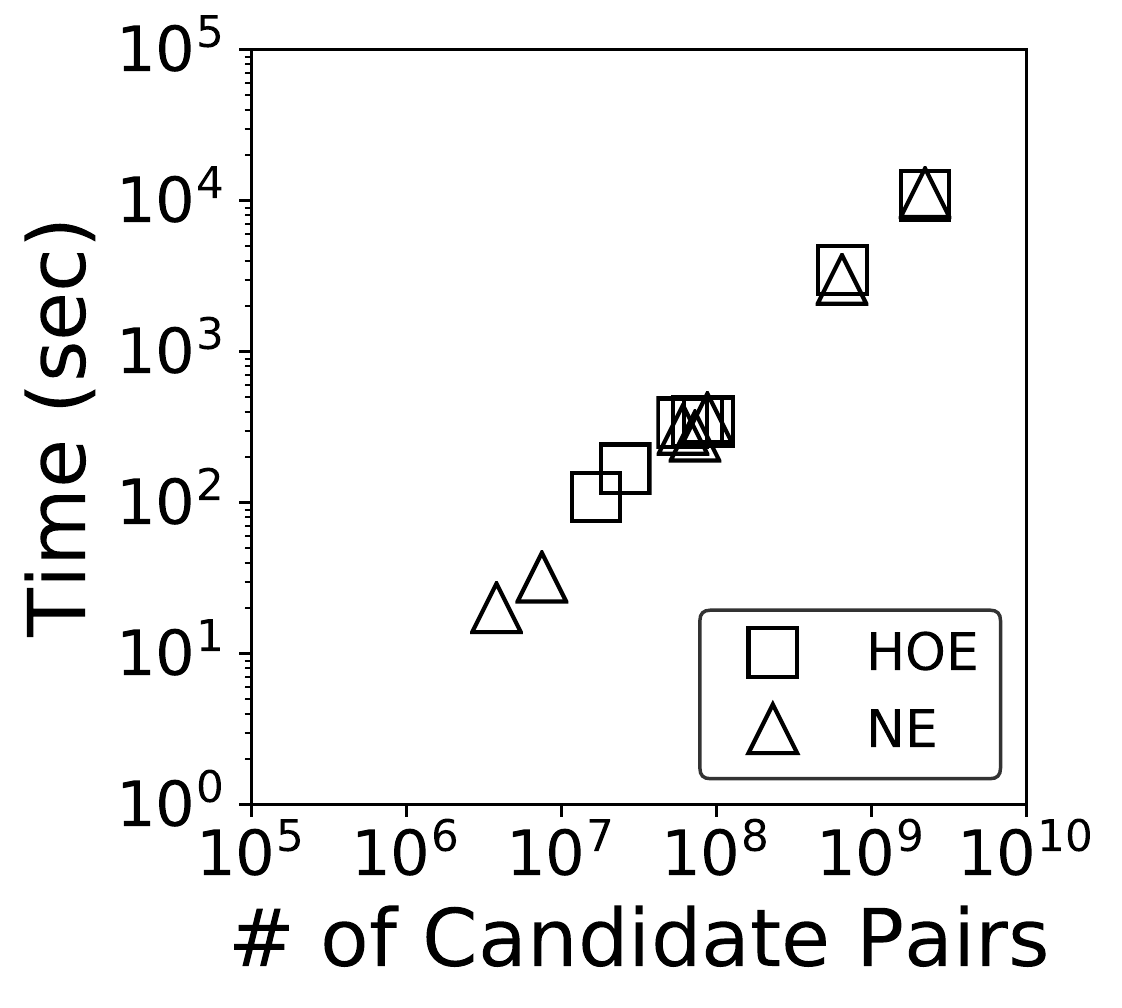}
        \vspace{-4mm}
        \label{fig:timecand}
        \caption{\label{expfig:timecand}{\small Linear dependence of the matching time on the number of candidate pairs.}}
    \end{subfigure}
    \end{minipage}
    \caption{\label{expfig:efficiency}\bf{Matching time and Number of candidate pairs.}}

    \vspace*{-1mm}
\end{figure*}

\begin{table}[!tb]
\large
\begin{center}
\vspace*{-2mm}\caption{\label{table:avgtime}\bf{Matching Time (ms) per Correct Match.}}
\vspace{-2mm}
\renewcommand{\arraystretch}{1}
\resizebox{\columnwidth}{!}{%
    \begin{tabular}{c|ccccccc}
    \noalign{\hrule height 1pt}
        Alg.  & Twitter & Dblp & Amazon & AU & SU & Youtube & WikiTalk  \\
        \hline
        \hline
        HOE  & 4.88 & 0.90 & 14.83 & 7.10 & 15.41 & 12.10 &  \textbf{40.32} \\
        NE  & \textbf{4.87} & \textbf{0.38} &  \textbf{0.51} & \textbf{4.67} & 12.55 & \textbf{11.89} &{42.02}\\
        EWS  & 34.72 & 3.72 & 17.43 & 233.1 & 208.1 & / & / \\
        GSANA& 1123 & 93.84 & 48.17 & 974.0 & 1875 & / & / \\
        AE   & 196.3 & 2.38 & 0.70 & 2116 & 3237 & / & / \\
        FRUI & 81.23 & 775.1 & 280.4 & 4.78 & \textbf{10.35} & 13926 & 6530 \\
    \noalign{\hrule height  1pt}
    \end{tabular}
}
\vspace*{-6mm}
\end{center}
\end{table}

\vspace{0.1cm}
\noindent
\textbf{2. Running time.}
The efficiency is of great concern in matching large graphs. 
In all competing algorithms, a basic operation is to compute some type of matching scores (defined differently in each method) of selected vertex pairs. Thus, in each algorithm, we call a vertex pair a \emph{candidate pair} if its matching score is ever computed.
The total matching time and the number of candidate pairs of all algorithms are presented in Figure~\ref{expfig:efficiency}.
The EWS, GSANA, and AE algorithms are too slow to finish the experiments on Youtube and WikiTalk datasets (within 6 hours). We make the following observations.
(1) In terms of total time, our two algorithms are very competitive; the FRUI method is the fastest in most cases, however it has low recalls, hence the total matching time is not a good indicator of efficiency.
Considering the recalls of different algorithms differ greatly, we use \emph{time per correct match}, i.e., total running time divided by the number of correct matches, as the metric for comparing time efficiency\footnote{This actually favors low-recall algorithm as it is getting harder and harder to identify more correct matches.}.
We observe NE is the most efficient among all algorithms (see Table~\ref{table:avgtime}), while HOE is also very competitive; our algorithms are typically 10 to 100 times faster in terms of time per correct match.
(2) For each algorithm (except for GSANA, whose running time is dominated by the embedding process), the more candidate pairs are considered, the longer matching time it takes.
(3) NE and HOE often consider fewer candidate pairs than other algorithms, because they directly utilize the higher-order structural information, which helps to identify true matches without generating a large set of candidates.
(4) The matching time of HOE and NE is almost linear in the number of candidate pairs, which verifies our theoretical analysis (see Figure~\ref{expfig:timecand}, each dot represents the result of HOE or NE on one dataset).

\begin{figure}[!tb]
    \centering
    \centering
    \captionsetup[subfigure]{oneside,margin={7mm,0cm}}
    \includegraphics[width=\linewidth]{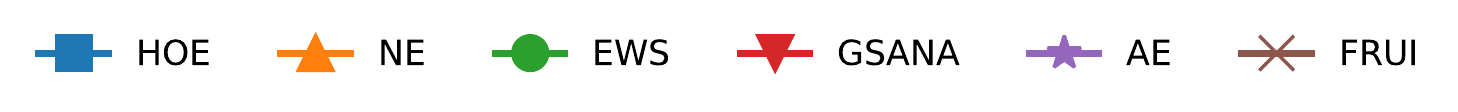}\label{fig:legend4}
    \vspace*{-5mm}

    \begin{subfigure}{0.22\textwidth}
        \includegraphics[width=\linewidth+0.1in]{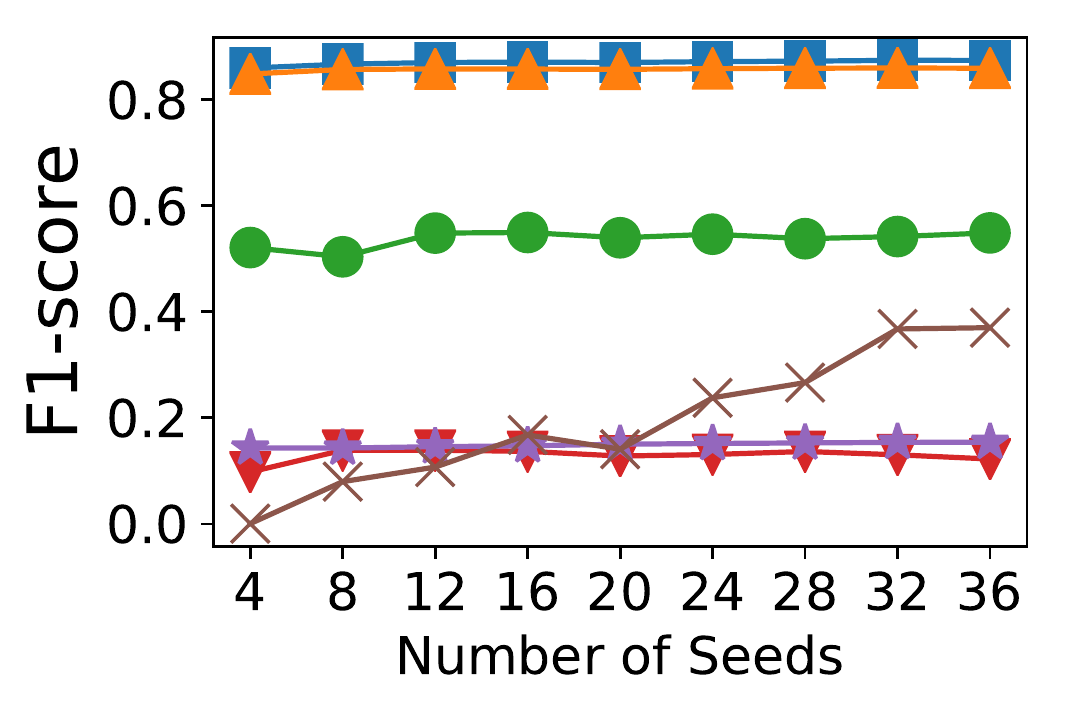}
        \vspace{-5mm}
        \caption{Twitter.}
        \label{expfig:twsnf1}
    \end{subfigure}
    \hspace{4mm}
    \begin{subfigure}{0.22\textwidth}
        \includegraphics[width=\linewidth+0.1in]{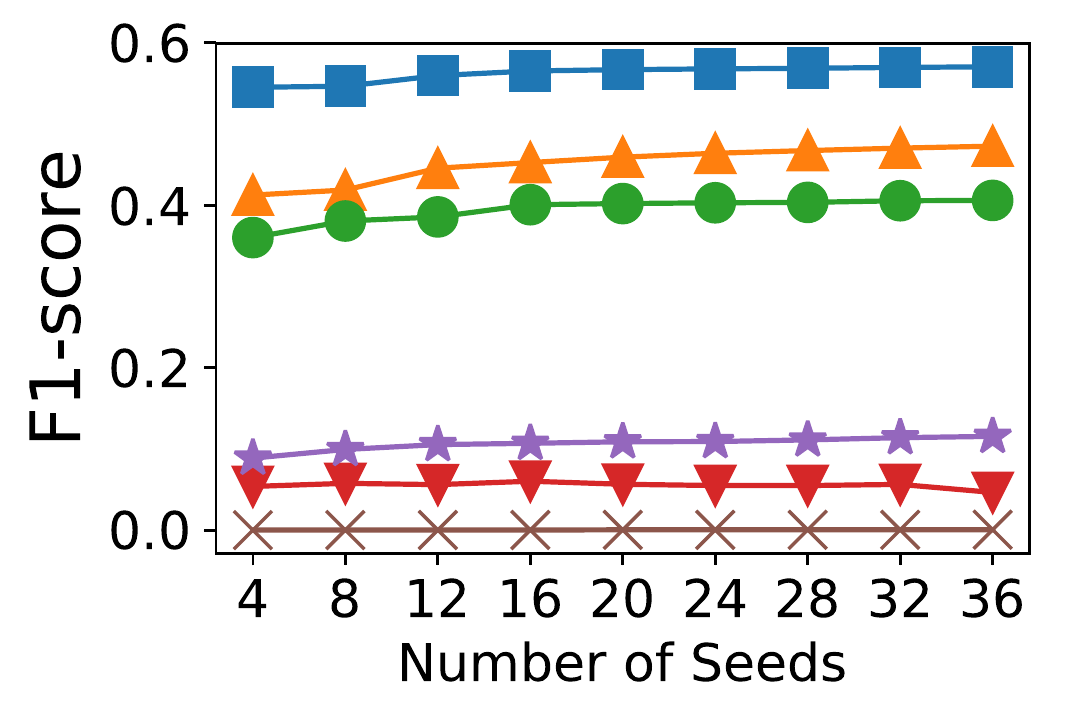}
        \vspace{-5mm}
        \caption{Dblp.}
        \label{expfig:dbsnf1}
    \end{subfigure}
    \vspace*{-3mm}
    \caption{\label{expfig:snf1}\bf{F1-score vs. Number of seeds.}}
    \vspace{-2mm}
\end{figure}

\begin{figure}[!tb]
    \centering
    \includegraphics[width=\linewidth]{expfig/legend4.pdf}\label{fig:legend4}
    \captionsetup[subfigure]{oneside,margin={7mm,0cm}}

    \vspace*{-1mm}
    \begin{subfigure}{0.22\textwidth}
        \includegraphics[width=\linewidth+0.1in]{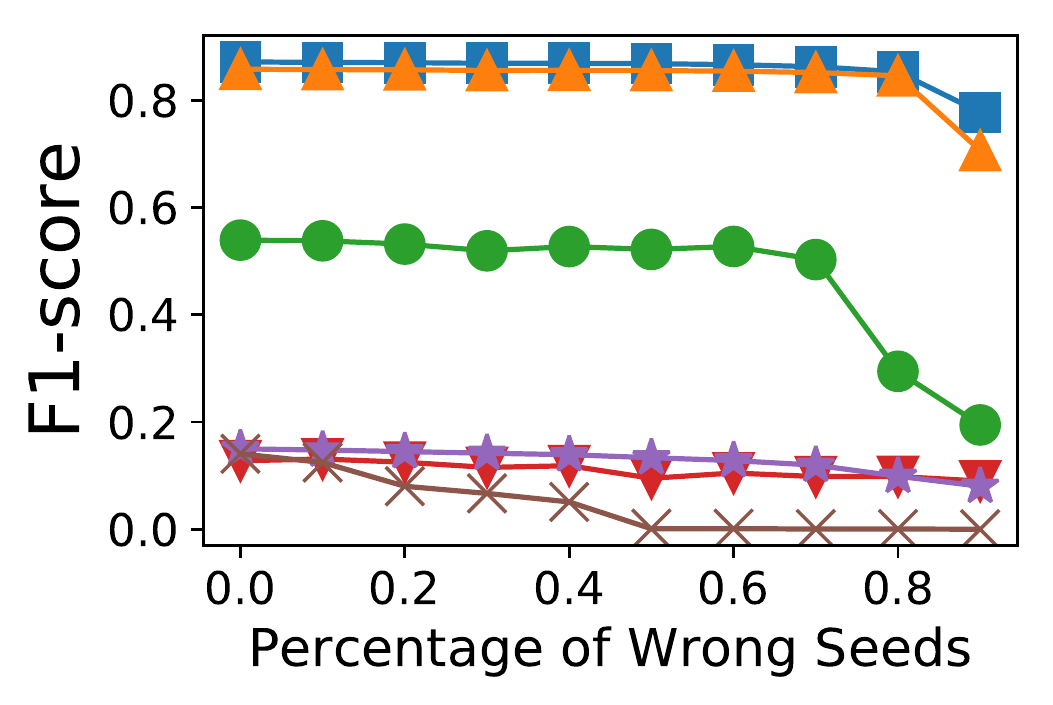}
        \vspace{-5mm}
        \caption{Twitter.}
        \label{fig:twwrf1}
    \end{subfigure}
    \hspace{4mm}
    \begin{subfigure}{0.22\textwidth}
        \includegraphics[width=\linewidth+0.1in]{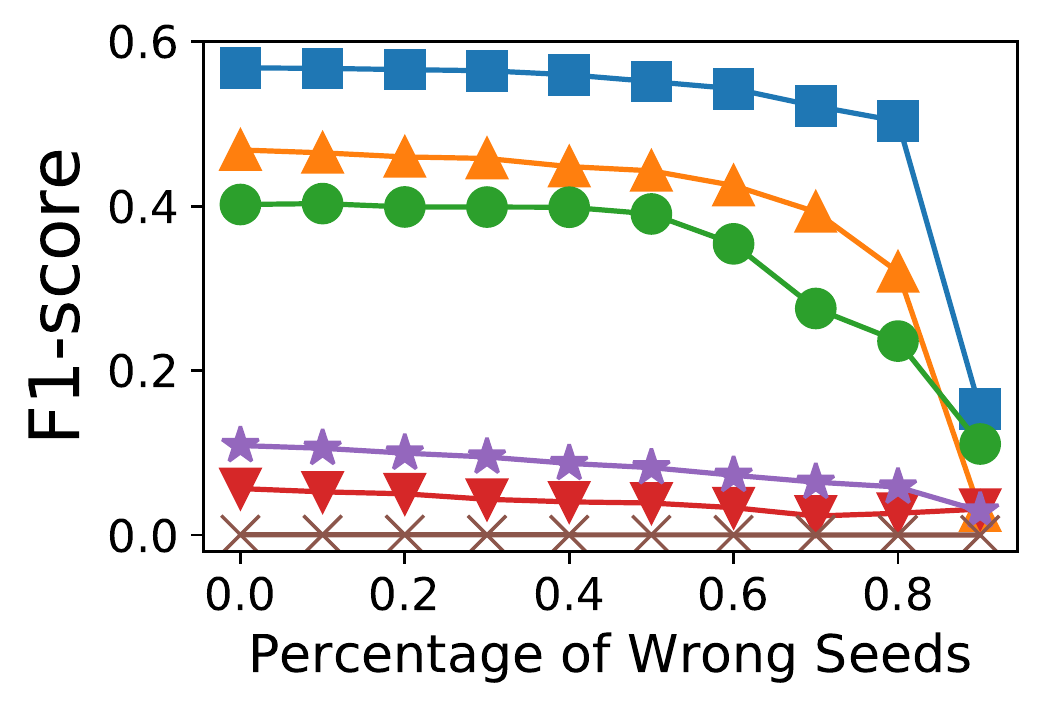}
        \vspace{-5mm}
        \caption{Dblp.}
        \label{fig:dbwrf1}
    \end{subfigure}

    \vspace*{-3mm}
    \caption{\label{expfig:wnf1}\bf{F1-score vs. Percentage of wrong seeds.}}
    \vspace{-3mm}
\end{figure}

\vspace{0.1cm}
\noindent
\textbf{3. Varying number of seeds.}
We then evaluate the F1-scores of algorithms with varying number of seeds (see Figure~\ref{expfig:snf1}). It is observed that both HOE and NE achieve high F1-scores even with only $4$ seeds on datasets Twitter and Dblp, and the performances are slightly improved as more seeds are provided. 
The FRUI method, a representative of traditional PGM algorithms which uses the neighboring information only, is highly sensitive to the number of seeds, and often terminates early with insufficient number of seeds. Experimental results on other datasets will appear in the full version due to space constraints.



\vspace{0.1cm}
\noindent
\textbf{4. Robustness against wrong seeds.}
In Figure~\ref{expfig:wnf1}, we measure the robustness of matching algorithms against wrong seeds. 
The matching algorithms are provided with $20$ seeds with the percentage of wrong seeds varying from $0$ to $0.9$.

%
%

It is observed that the wrong seeds have limited effect on the F1-score of HOE and NE, even if the percentage of wrong seeds is very high ($4$ correct seeds against $16$ wrong seeds, shown in Figure~\ref{expfig:wnf1}(a)); the performance on Dblp dataset is slightly more sensitive to the number of wrong seeds.
Such a phenomenon is also observed and explained in~\cite{kazemi2015growing} for EWS. From our experimental results, PPRGM is usually more robust than EWS.
On both datasets, the F1-score of EWS decreases significantly when the percentage of wrong seeds exceeds $0.6$.
The superior robustness of PPRGM is mainly due to the global impact of initial seeds: assume each wrong seed pairs is independently sampled, then in expectation any pair of two vertices $[u,v]$ receive similar PPR values from this wrong seed, and thus the set of random wrong seeds merely shifts the scores of all pairs by a same constant (in expectation), which doesn't affect the relative order of all matching scores.

\subsection{Study on the Parameters of PPRGM}\label{53performancestudy}

\begin{figure}[!tb]
    \centering
    \includegraphics[width=\linewidth]{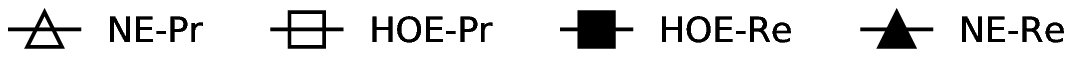}\label{fig:legend5}
    \captionsetup[subfigure]{oneside,margin={5mm,0cm}}

    \vspace{-1mm}
    \hspace*{-6mm}
    \begin{subfigure}{0.22\textwidth}
        \includegraphics[width=\linewidth+0.25in]{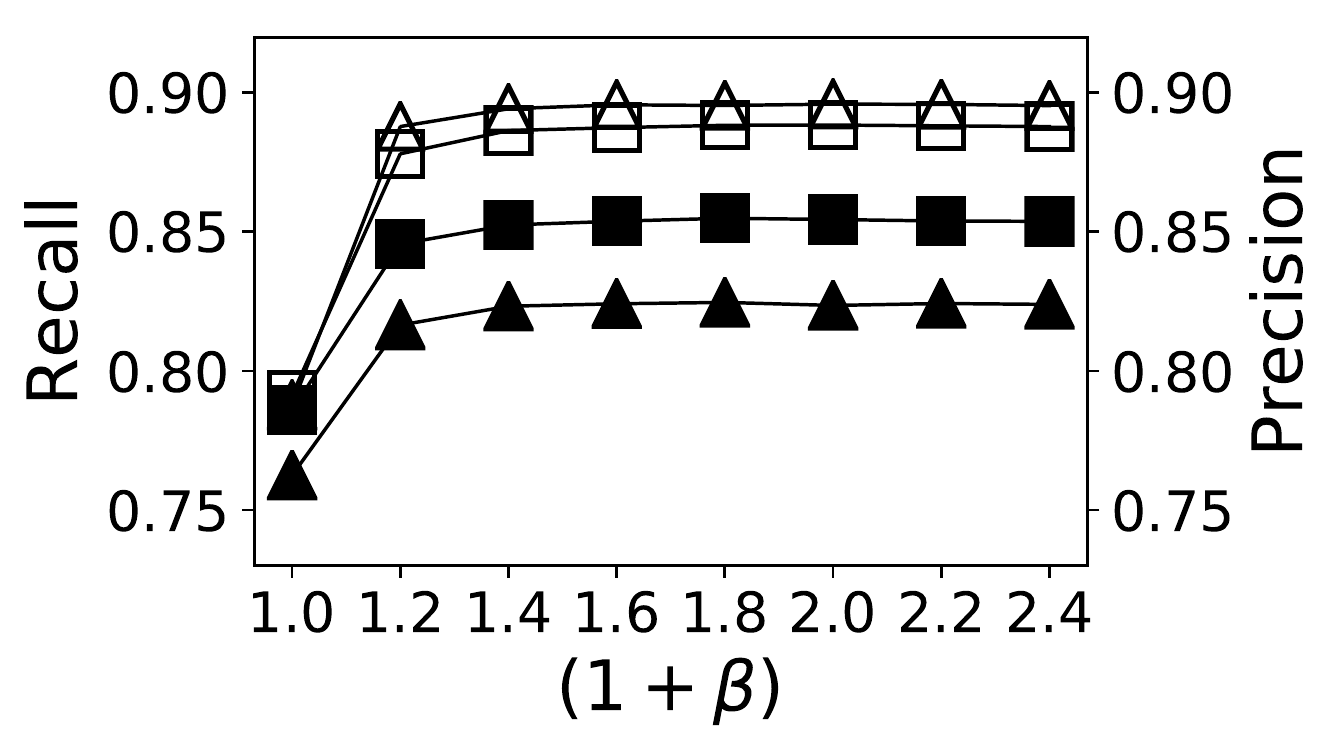}
        \vspace{-5mm}
        \caption{Twitter.}
        \label{fig:twbeta}
    \end{subfigure}
    \hspace{0.2in}
    \begin{subfigure}{0.22\textwidth}
        \includegraphics[width=\linewidth+0.15in]{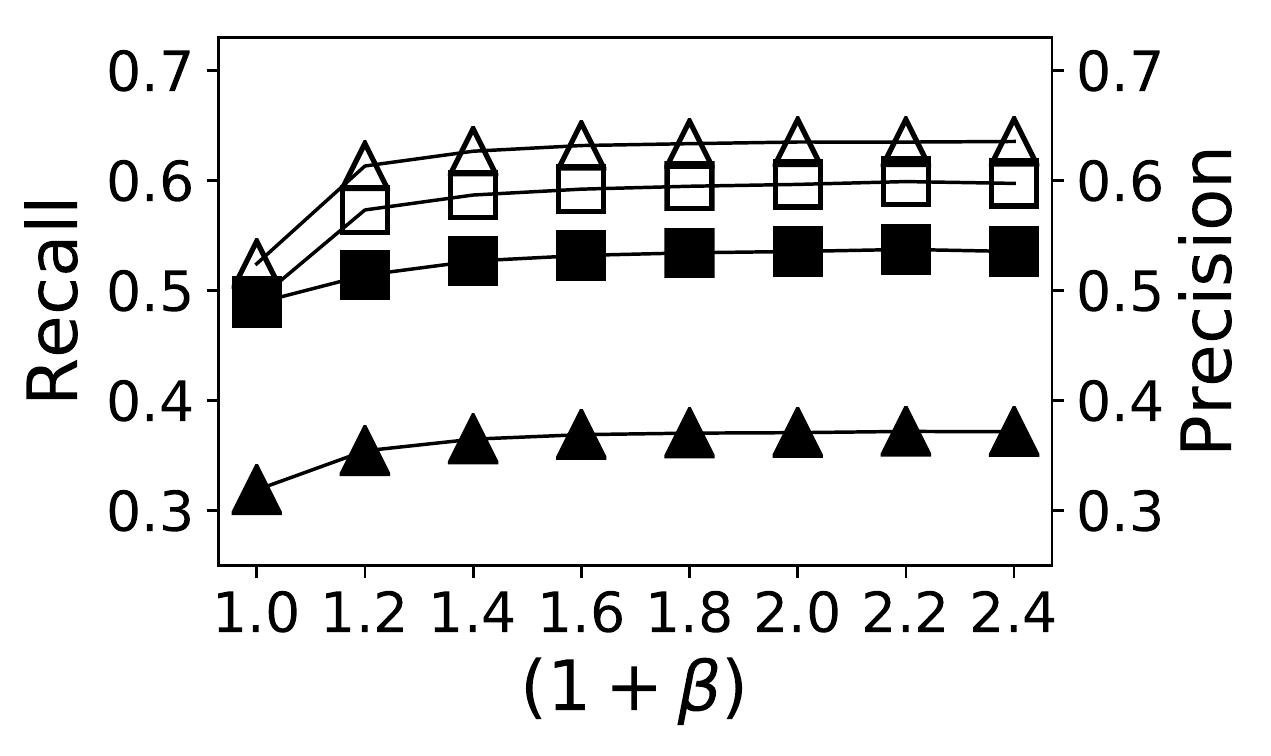}
        \vspace{-5mm}
        \caption{Dblp.}
        \label{fig:dbbeta}
    \end{subfigure}

    \vspace*{-3mm}
    \caption{\label{expfig:beta}\bf{Recall and Precision vs. $\beta$.}}
    \vspace{-2mm}
\end{figure}

\vspace{0.1cm}
\noindent
\textbf{1. Varying $\beta$ in postponing strategy.}
We evaluate the effect of postponing strategy by using different initial $\beta$ (see Figure~\ref{expfig:beta}).
Recall that PPRGM initializes a high $\beta$ ($\beta=1$), and decreases $\beta$ (relax the matching criteria) when there are no candidate pairs satisfying the matching criteria, until $\beta<\epsilon$ where $\epsilon$ is a very small constant.
It is observed the adopting postponing strategy (by setting $(1+\beta)\ge 1.2$) brings significant boost in recall and precision to both HOE and NE; the performances of HOE and NE are slightly improved with larger $\beta$ set at the beginning of algorithm. Overall, the performance is not sensitive to the setting of $\beta$; it is safe to set $1+\beta=2$ as in our default setting.

\begin{figure}[!tb]
    \centering
    \includegraphics[width=\linewidth]{expfig/legend5.eps}\label{fig:legend5}
    \captionsetup[subfigure]{oneside,margin={5mm,0cm}}

    \vspace{-1mm}
    \hspace*{-6mm}
    \begin{subfigure}{0.22\textwidth}
        \includegraphics[width=\linewidth+0.25in]{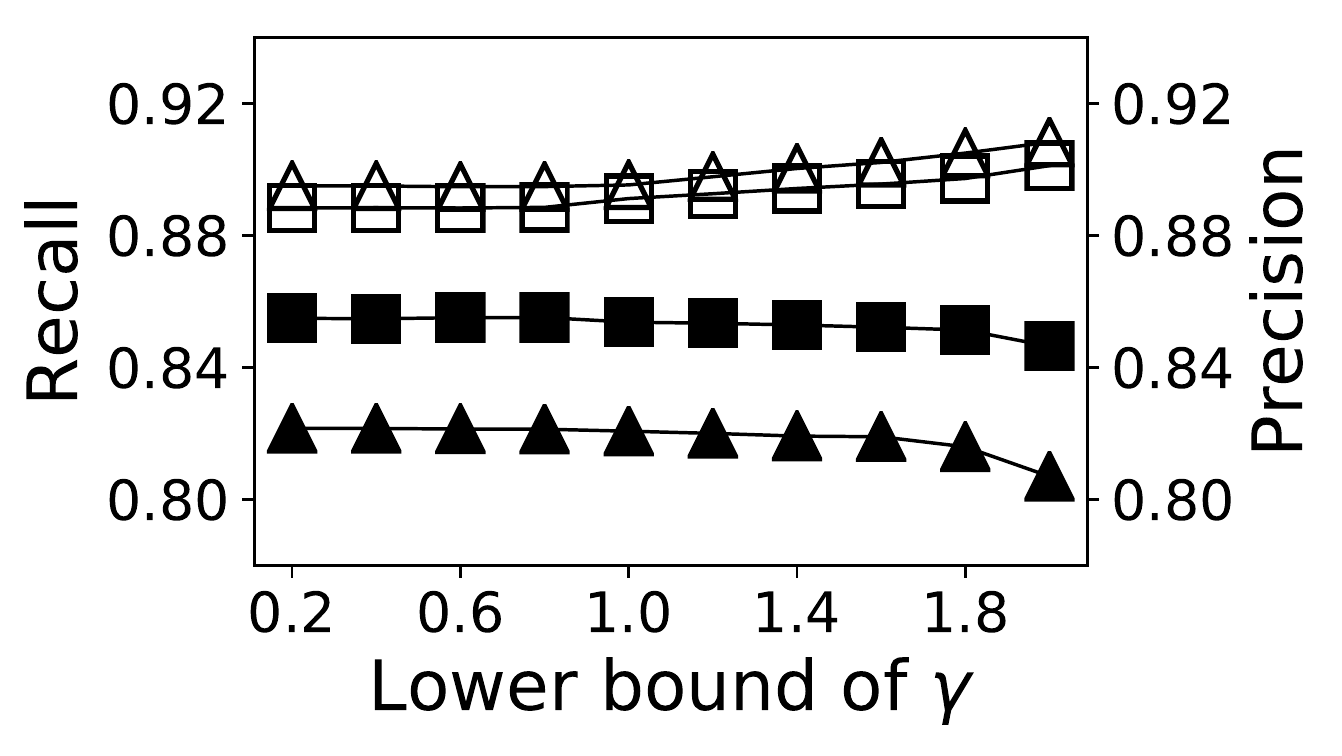}
        \vspace{-5mm}
        \caption{Twitter.}
        \label{fig:twgamma}
    \end{subfigure}
    \hspace{0.2in}
    \begin{subfigure}{0.22\textwidth}
        \includegraphics[width=\linewidth+0.15in]{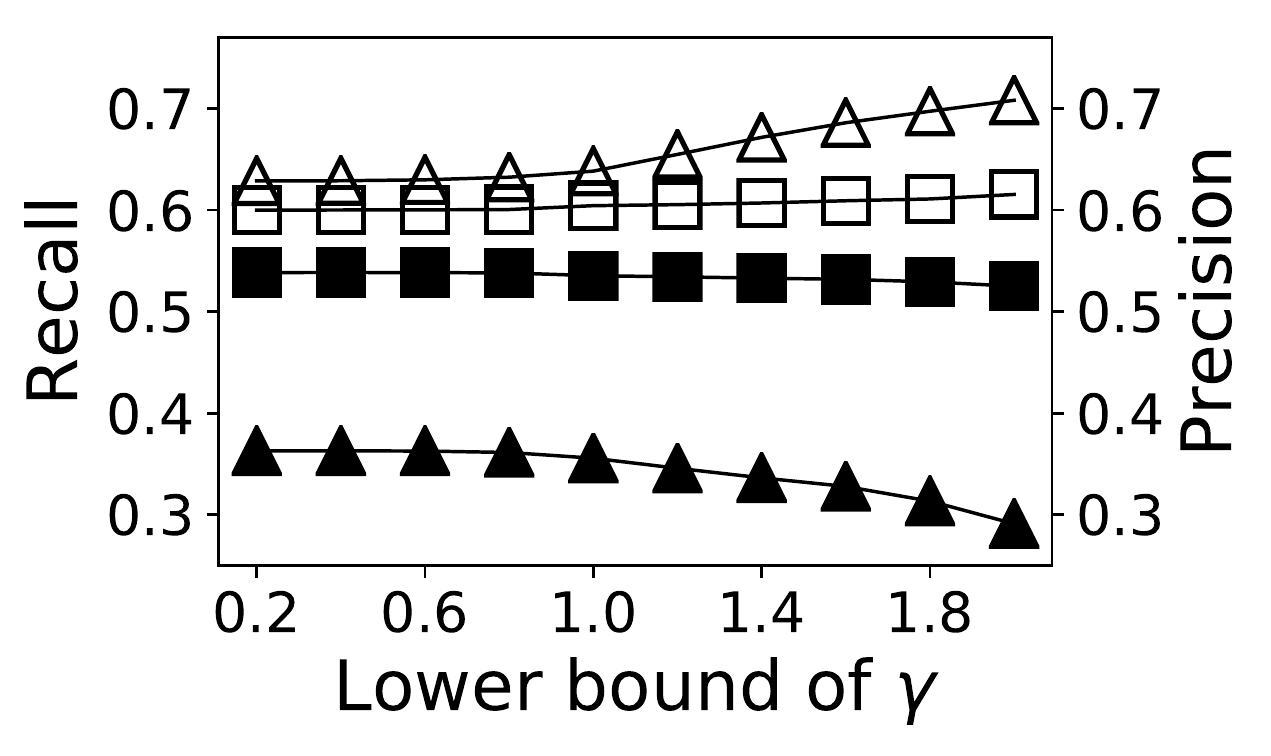}
        \vspace{-5mm}
        \caption{Dblp.}
        \label{fig:dbgamma}
    \end{subfigure}

    \vspace*{-3mm}
    \caption{\label{expfig:gamma}\bf{Recall and Precision vs. Lower bound of $\gamma$.}}
    \vspace{-2mm}
\end{figure}

\vspace{0.1cm}
\noindent
\textbf{2. Varying lower bound of $\gamma$ in postponing strategy.}
Similar as $\beta$, the choice of the initial value of $\gamma$ doesn't affect the performance by much, thus we only present the effect of varying lower bounds of $\gamma$ (see Figure~\ref{expfig:gamma}).
Recall that PPRGM initializes a high $\gamma$ ($\gamma=|S|/2$), and decreases $\gamma$ when there is no candidate pair satisfying the current matching criteria, until $\gamma$ reaches a lower bound (which is $1$ in our default setting). Therefore, a larger lower bound should reduce the recall since more strict criteria is adopted in the end, and thus it should also increase the precision. This is exactly what we have observed in the experiments.   In general, one could set the lower bound to $1$ unless a very high precision is needed.

\begin{figure}[!tb]
    \centering
    \includegraphics[width=\linewidth]{expfig/legend5.eps}\label{fig:legend5}
    \captionsetup[subfigure]{oneside,margin={5mm,0cm}}

    \vspace{-1mm}
    \hspace*{-6mm}
    \begin{subfigure}{0.22\textwidth}
        \includegraphics[width=\linewidth+0.25in]{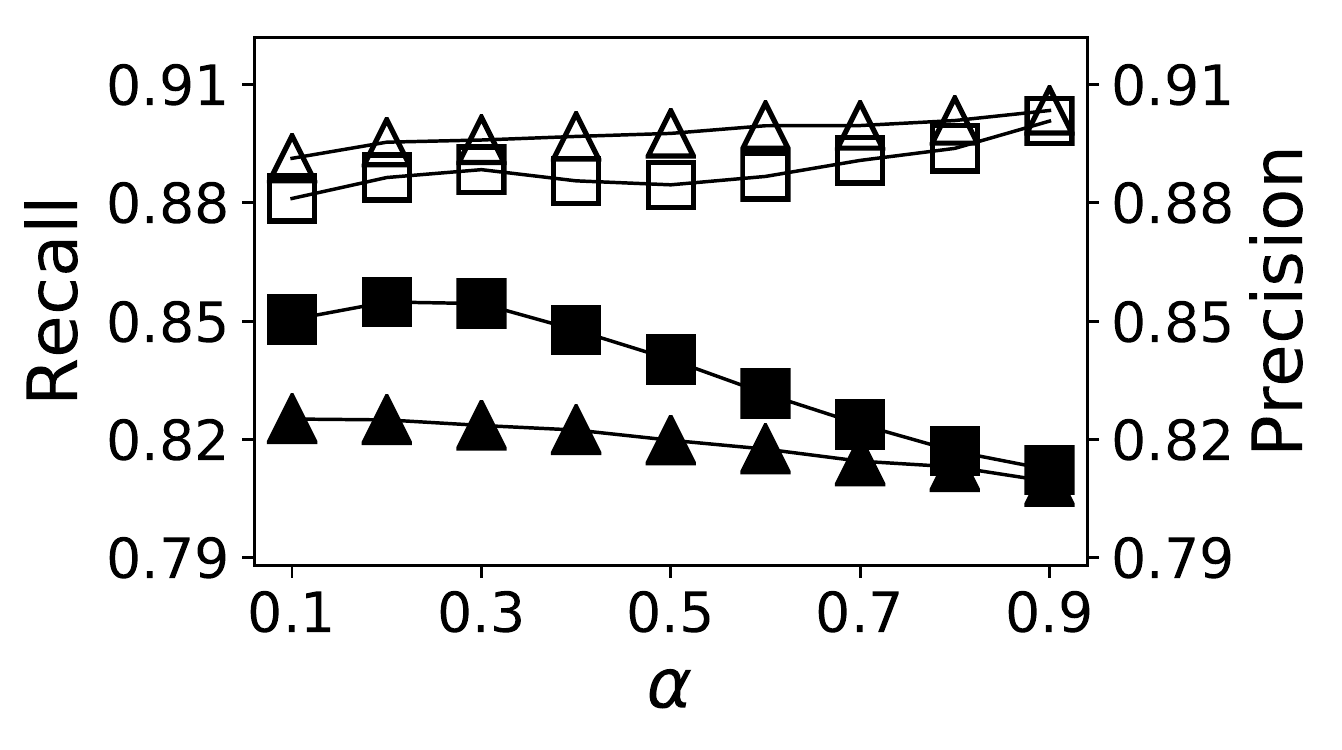}
        \vspace{-5mm}
        \caption{Twitter.}
        \label{fig:twalpha}
    \end{subfigure}
    \hspace{0.2in}
    \begin{subfigure}{0.22\textwidth}
        \includegraphics[width=\linewidth+0.15in]{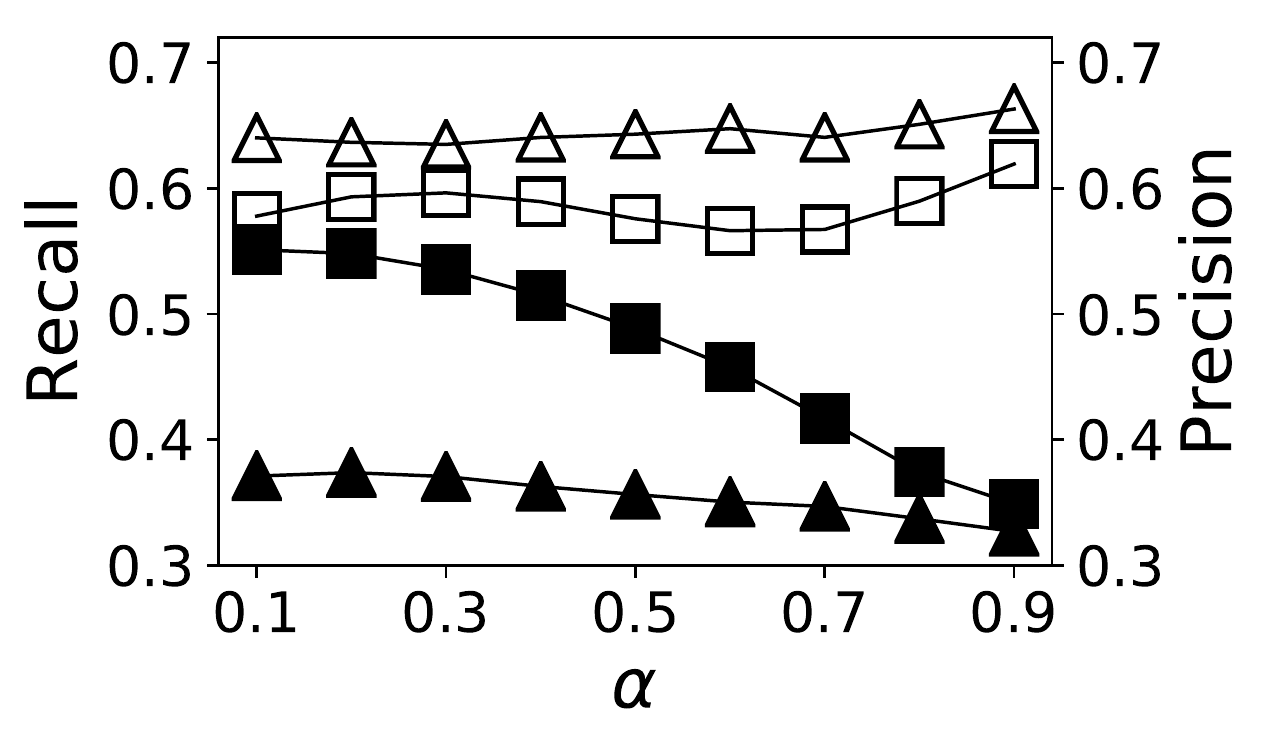}
        \vspace{-5mm}
        \caption{Dblp.}
        \label{fig:dbalpha}
    \end{subfigure}

    \vspace*{-3mm}
    \caption{\label{expfig:alpha}\bf{Recall and Precision vs. $\SP$.}}
    \vspace{-2mm}
\end{figure}

\vspace{0.1cm}
\noindent
\textbf{3. Varying $\SP$ in Forward-Push.}
The choice of the stopping probability in the decaying random walk in PageRank varies in different applications, which is largely based on empirical studies. A discussion has been presented in~\cite{gleich2010inner}.
We evaluate the performance of HOE and NE with varying stopping probabilities $\SP$ in the decaying random walk (see Figure~\ref{expfig:alpha}).
We observe that the recall of both HOE and NE increases as $\SP$ becomes smaller (the boost is more significant on HOE),
because a Forward-Push with smaller $\SP$ may infect a larger range of vertices and thus will include more candidates.
For precision, varying $\SP$ doesn't affect the results a lot. The optimal performance in terms of F1-score appears when $\alpha\approx 0.3$.
Such phenomena are also observed on other datasets, which will be provided in the full version, so our experimental study empirically concludes that $\alpha=0.3$ is the right choice for seed graph matching.


\begin{figure}[!tb]
    \centering
    \captionsetup[subfigure]{oneside,margin={5mm,0cm}}
    \hspace*{-5mm}
    \begin{subfigure}{0.22\textwidth}
        \includegraphics[width=\linewidth+0.2in]{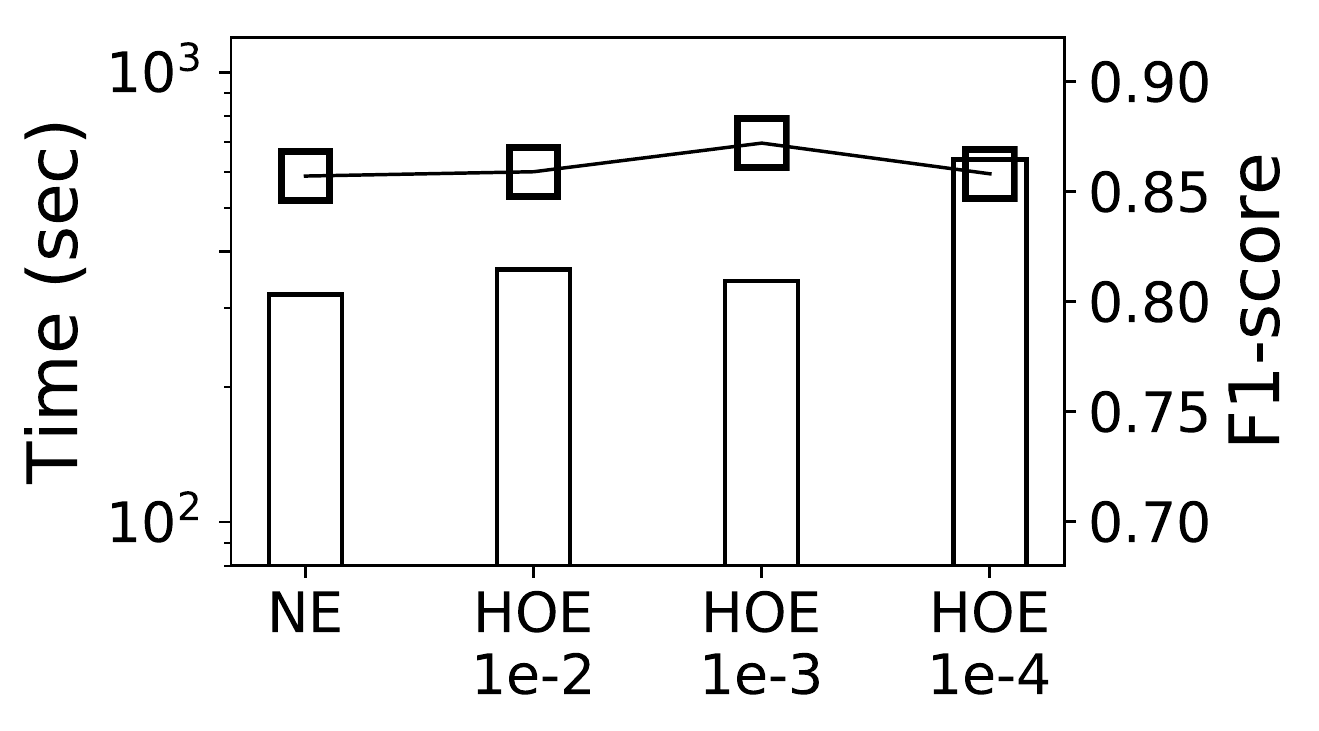}
        \vspace{-5mm}
        \caption{Twitter.}
        \label{fig:twrmax}
    \end{subfigure}
    \hspace{5mm}
    \begin{subfigure}{0.22\textwidth}
        \includegraphics[width=\linewidth+0.2in]{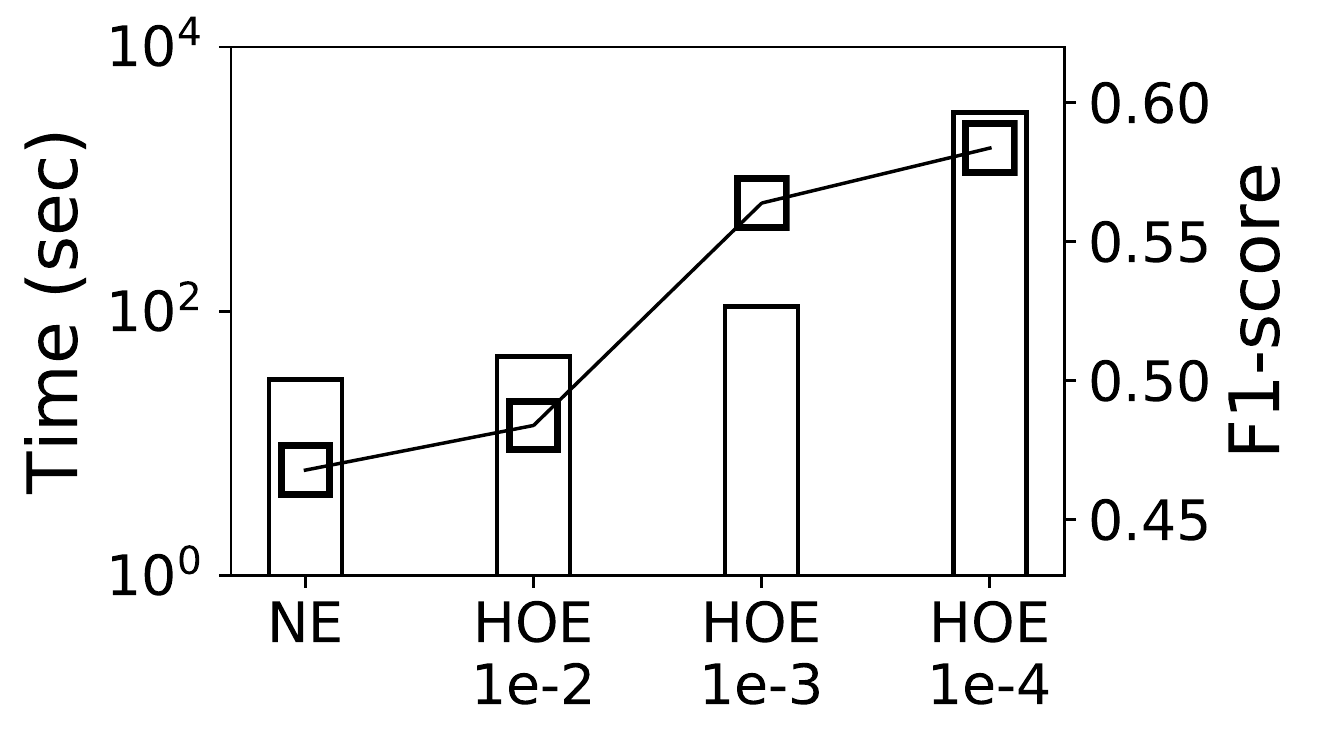}
        \vspace{-5mm}
        \caption{Dblp.}
        \label{fig:dbrmax}
    \end{subfigure}

    \vspace*{-3mm}
    \caption{\label{figure:8}\bf{F1-score and Number of candidate pairs vs. $r'_{max}$.}}
    \vspace{-2mm}
\end{figure}

\vspace{0.1cm}
\noindent
\textbf{4. Varying $r'_{max}$ in Forward-Push.}
We finally compare the performances of NE and HOE with varying $r'_{max}$. 
We record the matching times (represented by bars) and the F1-scores (represented by lines) of different settings.
Firstly, the running time increase as smaller $r'_{max}$ is used. This matches the analysis in Lemma~\ref{lemma:time}.
Secondly, for different datasets,  the optimal $r'_{max}$ varies, e.g., on the Dblp graph, HOE with a small $r'_{max}$ obtains noticeable boost in F1-score (Figure~\ref{figure:8}(b)). However, on Twitter (Figure~\ref{figure:8}(a)), the F1-score of HOE is almost oblivious to the setting of $r'_{max}$.
We observe that the preference of $r'_{max}$ on different datasets is correlated with the average distance of the graph.
In matching graphs with large average distances (e.g., Dblp, Amazon), a small $r'_{max}$ could boost the F1-score by a noticeable margin. Contrarily, when matching two graphs with small average distances (e.g., Twitter, Wiki, AskUbuntu, etc.), HOE with a large $r'_{max}$ or even NE could get high-quality results.

\subsection{Experiment Summary}
(1) HOE and NE have overall higher recalls and precisions than the start-of-the-art algorithms. 
(2) HOE and NE are typically more scalable than previous methods and NE is the most efficient one in terms of time per correct match.
(3) HOE and NE only require very few seeds to achieve peak precisions and recalls; HOE and NE are very robust against wrong seed pairs.
(4) The setting of all the parameters in PPRGM (except for $r'_{max}$) is not affected by the underlying dataset, which greatly eases the tuning of parameters.
(5) HOE with a small $r'_{max}$ is recommended for matching graphs with large average distances, while NE (or HOE with a large $r'_{max}$) is recommended for graphs with small average distances.

\section{Conclusions}

In this paper, we propose a powerful seeded graph matching framework by directly using the high order structural information in the graph.
We propose to quantify the connection between matched vertices and vertices to be matched with personalized PageRank.
A score function is defined to compute the matching score of a pair of vertices based on their PPR values w.r.t. the initial seeds and early matches.
Several optimization strategies are proposed to further boost the performance of our PPRGM framework.
Extensive experiments on large-scale real graphs demonstrate that our new seeded graph matching algorithms outperform start-of-the-art algorithms w.r.t. quality, efficiency, and robustness by noticeable margins.


\balance
\bibliographystyle{abbrv}
\bibliography{bibliography}
\end{document}